\documentclass[journal,twocolumn]{IEEEtran}

\IEEEoverridecommandlockouts

\usepackage[cmex10]{amsmath} 
\usepackage{amsfonts,stmaryrd,acronym,amssymb,amsthm,dsfont}
\usepackage{graphicx}
\usepackage{acronym}
\usepackage{url}
\usepackage{cite}
\usepackage{hyperref}
\usepackage{float}
\usepackage{dblfloatfix}
\usepackage{flushend}
\usepackage{enumerate}
\usepackage{lipsum}
\usepackage{diagbox}
\usepackage[dvipsnames]{xcolor}
\usepackage{mathrsfs}
\usepackage{bm}
\usepackage{mathtools}

\graphicspath{{graphics/}}

\newtheorem{theorem}{Theorem}
\newtheorem{corollary}{Corollary}
\newtheorem{definition}{Definition}

\newtheorem{remark}{Remark}

\newtheorem{example}{Example}
\newtheorem{proposition}{Proposition}

\newtheorem{lemma}{Lemma}[]

\makeatletter
\def\blfootnote{\xdef\@thefnmark{}\@footnotetext}
\makeatother
\newcommand{\delf}[1]{\ensuremath{\mathds{1}}}
\newcommand{\indic}[1]{\ensuremath{\mathds{1}}}
\newcommand{\abs}[1]{\ensuremath{\left|#1\right|}}   

\newcommand{\eqdef}{\ensuremath{\triangleq}}   
\newcommand{\sbra}[2]{\ensuremath{[{#1}{\,:\,}{#2}]}}%

\newcounter{mytempeqcounter}

\allowdisplaybreaks
\allowbreak

\makeatletter%
\if@twocolumn%
\else
\fi%
\makeatother%

\newcommand{\calA}{\mathcal{A}}

\newcommand{\calB}{\mathcal{B}}

\newcommand{\calD}{\mathcal{D}}

\newcommand{\calE}{\mathcal{E}}
\newcommand{\bbE}{\mathbb{E}}

\newcommand{\calF}{\mathcal{F}}
\newcommand{\bbF}{\mathbb{F}}

\newcommand{\calH}{\mathcal{H}}

\newcommand{\calI}{\mathcal{I}}

\newcommand{\calJ}{\mathcal{J}}

\newcommand{\calM}{\mathcal{M}}

\newcommand{\bbN}{\mathbb{N}}

\newcommand{\bbP}{\mathbb{P}}

\newcommand{\calQ}{\mathcal{Q}}

\newcommand{\bbR}{\mathbb{R}}

\newcommand{\calT}{\mathcal{T}}

\newcommand{\calX}{\mathcal{X}}

\newcommand{\calY}{\mathcal{Y}}

\DeclareMathOperator{\C}{C}

\newcommand{\Squad}{\hspace{0.5em}}
\newcommand{\ts}{\textsuperscript}
\acrodef{ACDIS}[ACDIS]{Adaptive Communication Decision and Information Systems}
\acrodef{AEP}{Asymptotic Equipartition Property}
\acrodef{AoA}{Angle of Arrival}
\acrodef{AWGN}{Additive White Gaussian Noise}
\acrodef{AVC}[AVC]{Arbitrarily Varying Channel}
\acrodefplural{AVC}{Arbitrarily Varying Channels}
\acrodef{PIR-PNSI}{Private Information Retrieval with Private Noisy Side Information}
\acrodef{BER}{Bit-Error-Rate}
\acrodef{BEC}{Binary Erasure Channel}
\acrodefplural{BEC}{Binary Erasure Channels}
\acrodef{BSC}{Binary Symmetric Channel}
\acrodefplural{BSC}{Binary Symmetric Channels}
\acrodef{BPSK}{Binary Phase-Shift Keying}
\acrodef{BICM}[BICM]{Bit-Interleaved Coded-Modulation}
\acrodef{CDF}[CDF]{Cumulative Distribution Function}
\acrodef{CGF}[CGF]{Cumulant Generating Function}
\acrodef{CLT}[CLT]{Central Limit Theorem}
\acrodef{CSI}[CSI]{Channel State Information}
\acrodef{DMC}[DMC]{Discrete Memoryless Channel}
\acrodef{DMS}[DMS]{Discrete Memoryless Source}
\acrodef{ERM}[ERM]{Empirical Risk Minimization}
\acrodef{FER}[FER]{Frame Error Rate}
\acrodef{ICA}[ICA]{Independent Component Analysis}
\acrodef{iid}[i.i.d.]{independent and identically distributed}
\acrodef{IoT}[IoT]{Internet of Things}
\acrodef{KKT}[KKT]{Karush-Kuhn Tucker}
\acrodef{LASSO}[LASSO]{Least Absolute Shrinkage and Selection Operator}
\acrodef{LPD}[LPD]{Low Probability of Detection}
\acrodef{LDPC}[LDPC]{Low-Density Parity-Check}
\acrodef{LLMS}[LLMS]{Linear Least Mean Square}
\acrodef{LMS}[LMS]{Least Mean Square}
\acrodef{MAC}[MAC]{multiple-access channel}
\acrodef{MGF}[MGF]{Moment Generating Function}
\acrodef{MLC}[MLC]{Multi-Level Coding}
\acrodef{MLE}[MLE]{Maximum Likelihood Estimate}
\acrodef{MIMO}[MIMO]{Multiple-Input Multiple-Output}
\acrodef{MISO}{Multiple-Input Single-Output}
\acrodef{MSD}[MSD]{Multi-Stage Decoding}
\acrodef{MMSE}[MMSE]{Minimum Mean-Square Error}
\acrodef{PAC}[PAC]{Probably Approximately Correct}
\acrodef{PCA}[PCA]{Principal Component Analysis}
\acrodef{PDF}[PDF]{Probability Density Function}
\acrodef{PMF}[PMF]{Probability Mass Function}
\acrodef{PPM}[PPM]{Pulse Position Modulation}
\acrodef{PSD}{Power Spectral Density}
\acrodef{PSK}{Phase Shift Keying}
\acrodef{QKD}{Quantum Key Distribution}
\acrodef{ROC}{Receiver Operating Characteristic}
\acrodef{CVQKD}{Continuous-Variable \ac{QKD}}
\acrodef{QPSK}{Quadrature Phase-Shift Keying}
\acrodef{RV}{random variable}
\acrodef{SIMO}{Single-Input Multiple-Output}
\acrodef{SNR}{Signal-to-Noise Ratio}
\acrodef{SVM}[SVM]{Support Vector Machine}
\acrodef{TPCP}{Trace-Preserving Completely-Positive}
\acrodef{wrt}[w.r.t.]{with respect to}
\acrodef{WSS}{Wide Sense Stationary}
\acrodef{RHS}{Right Hand Side}
\acrodef{LHS}{Left Hand Side}
\acrodef{PIR}{Private Information Retrieval}
\acrodef{MDS}{Maximum Distance Separable}

\begin{document}

\title{Private Information Retrieval with Private Noisy Side Information}

\author{
\IEEEauthorblockN{Hassan ZivariFard and R\'{e}mi A. Chou}\\
\thanks{H.~ZivariFard is with the Department of Electrical Engineering, Columbia University, New York, NY 10027. R.~Chou is with the Department of Computer Science and Engineering, The University of Texas at Arlington, Arlington, TX 76019. This work is supported in part by NSF grant CCF-2047913. E-mails: hz2863@columbia.edu and remi.chou@uta.edu. Part of this work was presented at the 2023 IEEE International Symposium on Information Theory \cite{ISIT23_PIR}.}
}
\maketitle
\date{}

\begin{abstract}
\label{sec:Abstract}
Consider \ac{PIR}, where a client wants to retrieve one file out of $K$ files that are replicated in $N$ different servers and the client selection must remain private when up to $T$ servers may collude. Additionally, suppose that the client has noisy side information about each of the $K$ files, and the side information about a specific file is obtained by passing this file through one of $D$ possible discrete memoryless test channels, where $D\le K$. While the statistics of the test channels are known by the client and by all the servers, the specific mapping $\boldsymbol{\calM}$ between the files and the test channels is unknown to the servers. We study this problem under two different privacy metrics. Under the first privacy metric, the client wants to preserve the privacy of its desired file selection and the mapping $\boldsymbol{\calM}$. Under the second privacy metric, the client wants to preserve the privacy of its desired file and the mapping $\boldsymbol{\calM}$ but is willing to reveal the index of the test channel that is associated to its desired file. For both of these two privacy metrics, we derive the optimal normalized download cost. Our problem setup generalizes \ac{PIR} with colluding servers, \ac{PIR} with private noiseless side information, and \ac{PIR} with private side information under storage constraints. 
\end{abstract}

\section{Introduction}
\label{sec:Intro}
\ac{PIR} refers to a problem where a client wishes to download, as efficiently as possible, one of the $K$ files that are replicated among a set of distributed servers such that the servers cannot learn anything about the client's file selection~\cite{PIR95,PIR98}. Aside from its direct applications in data security and privacy, it is closely related to many fundamental problems such as secret sharing~\cite{Secret_Sharing_2,beimel2011secret} and oblivious transfer~\cite {Oblivious98,rabin1981exchange}, which is also called symmetric \ac{PIR} and is a \ac{PIR} problem where the server wants to keep any non-selected file private from the client. Therefore, \ac{PIR} is a subject that relates to different areas in computer science. 

The \ac{PIR} problem was studied in~\cite{SunJafar171} from an information-theoretic point of view to characterize the maximum number of bits of desired information that can be retrieved privately per bit of downloaded information. In~\cite{SunJafar171}, the authors showed that this quantity is $(1+ 1/N+ 1/N^2+\dots+1/N^{K-1})^{-1}$ when a client wishes to retrieve one of the $K$ files that are distributed in $N$ replicated and non-colluding servers. This problem was subsequently extended to various scenarios. \cite{SunJafar172}~considered a \ac{PIR} problem where $T$ of the $N$ servers may collude and some of the servers may not respond. \cite{CodedPIR,CodedPIRDSS,CodedLi20}~studied \ac{PIR} with $N$ non-colluding servers, where each server stores an MDS-coded version of the $K$ files. \cite{SunJafar19,WangSkogland19}, extended the results to symmetric \ac{PIR}, in which the privacy of both the client and the servers is considered. 

\subsection{Overview of the setting studied in this paper}
In this paper, we study a \ac{PIR} problem where the client wants to retrieve one of the $K$ files that are replicated in $N$ servers and $T$ of these servers may collude. As reviewed in the next section,  only \ac{PIR} with noiseless side information has been studied in the literature, i.e., the client has access to a subset of the files or portions of each file and their corresponding positions in the original files. By contrast, in our problem setting, the client has a noisy version of each file which is obtained by passing each file through a discrete memoryless test channel. We assume that there are $D\le K$ different test channels whose statistics are public knowledge and known by the client and the servers.\footnote{
We assume that the statistics of the test channels, i.e. $C^{(\ell)}$, $\ell\in[D]$, are public information. Note that for each file, $X_k^n$, for $k\in[K]$, the client has side information about $X_k^n$ which can potentially be $\emptyset$, consequently, no more than $K$ test channels are needed to model the side information available at the client.
} We denote the mapping between the files and the test channels by $\boldsymbol{\calM}$. We study this problem under two different privacy metrics. 
For the first privacy metric, the client wants to keep the index of the desired file and the entire mapping $\boldsymbol{\calM}$ secret from the servers, and this includes the index of the test channel that is associated with the desired file. For the second privacy metric, the client wants to keep the index of the desired file and the mapping $\boldsymbol{\calM}$ secret from the servers, but the client is willing to reveal the index of the test channel that is associated with the desired file, i.e., $\boldsymbol{\calM}(Z)$.
For both privacy metrics, we derive the optimal normalized download cost, and we show that the second privacy metric always leads to a lower normalized download cost.

\subsection{Motivations}
\label{sec:Motivs}
Consider the following motivational example of \ac{PIR}, e.g., \cite{Private_Updating}: a stock market investor may want to privately retrieve some of the stock records because showing interest in one specific record could undesirably affect its value. Now, consider the case where an investor has already retrieved some or all of the stock records in the past. The investor could now retrieve a record by leveraging their knowledge of outdated records, which represents side information. As another example, the client could have acquired the noisy side information in several ways. For example, the user could have acquired a noisy version of the files opportunistically from other users in its network, overheard them from a wireless noisy broadcast channel, or downloaded them previously through classical \ac{PIR} schemes from other servers. Note that the availability of noisy side information encompasses having obtained parts of the files in a noiseless manner.
Also, the noise could have been the result of storing the files for a long period of time.

Note that in the stock market example, revealing the mapping between the stock records and the test channels shows how much information the investor has about each stock record, which is not in the interest of the investor since they do not want to affect the value of the stock records. Additionally, if the client has a subset of the files in a noiseless manner as side information \cite{PIR_SI_20,ChenWangJafar20}, then not keeping private the mapping between the test channels and the files may reveal to the servers the indices of the files that are available as side information at the client. These are examples of our first privacy metric. 
We also consider a privacy metric where the client reveals the index of the test channel associated with the desired file. As discussed in Example~\ref{ex:PNLSI} and Remark~\ref{EX:Comaring_Main_Results}, this privacy metric can lead to a lower download cost,  when, for example, the desired file is available in the side information in a noiseless manner and the client does not need to download anything.
\subsection{Related works}
As identified in~\cite{ChenWangJafar20}, three main models for \ac{PIR} with side information have been studied in the literature, which are summarized as follows.
\begin{itemize}
    \item \ac{PIR} with side information globally known by all the terminals: the effect of side information on the information-theoretic capacity of the \ac{PIR} problem was first studied in~\cite{Tandon17}, where the author considers a \ac{PIR} problem in which a client wishes to privately retrieve one out of $K$ files from $N$ replicated non-colluding servers. Specifically, in~\cite{Tandon17}, the client has a local cache that can store any function of the $K$ files. 
    \item \ac{PIR} with side information, where the privacy of the side information is not required: the single-server \ac{PIR} problem, where the client has access to a subset of the files and wants to protect only the identity of the desired file, is introduced and solved in~\cite{PIR_SI_20}. An achievability result for the multiserver case is also derived in~\cite{PIR_SI_20}, and was later shown to be optimal in~\cite{LiGastpar20}. Single-server \ac{PIR} when the client knows $M$ files out of $K$ files, or a linear combination of $M$ files, has further been studied in~\cite{Heidarzadeh18,Heidarzadeh19,PIRCodedSI} under various scenarios. Also, a multiserver \ac{PIR} when the client has a noisy version of the desired file is studied in~\cite{Private_Updating}. A more general notion of partial privacy for noiseless side information is introduced in \cite{Heidarzadeh22}, where a subset of the files available as side information is kept private from the servers while the complement of this subset of files is not required to be kept private.
    \item \ac{PIR} with private side information, where the joint privacy of the file selection and the side information is required: \cite{PIR_SI_20}~derived an achievable normalized download cost for $N$ replicated and non-colluding servers. \ac{PIR} from $N$ replicated and non-colluding servers, where a cache-enabled client possesses side information, in the form of uncoded portions of the files, that is unknown to the servers, is studied in~\cite{PIRStorageConstraint20}. Specifically, in~\cite{PIRStorageConstraint20}, the client knows the first $r_i$ bits, for $i\in[M]$, of $M$ randomly selected files, and the identities of these side information files need to be kept private from the servers. Also, \ac{PIR} from $N$ replicated and non-colluding servers when the client knows $M$ files out of $K$ files as side information, and each server knows the identity of a subset of the side information files, is studied in~\cite{PIRPatialSI19}. In~\cite{ChenWangJafar20}, the authors studied the \ac{PIR} problem where the client wishes to retrieve one of the $K$ files from $N$ replicated servers, when $T$ of the servers may collude, and the client has access to $M$ files in a noiseless manner. This problem is extended to the case where the client wants to retrieve multiple files privately in~\cite{Siavoshani21}. 
\end{itemize}

\textit{Difference between our model and previous models:} In this paper, we focus on \ac{PIR} with private side information. Note that the side information in the \ac{PIR} problems in~\cite{Tandon17,CodedPIRPrefetching,CodedPIRPrefetchingJSAC,PIRStorageConstraint20,PIR_SI_20,PIRPatialSI19,Heidarzadeh18,Heidarzadeh19,Heidarzadeh22,LiGastpar20,PIRCodedSI,ChenWangJafar20,Siavoshani21} is always noiseless, in the sense that all the side information available at the client corresponds to sub-sequences of each file and the client knows the corresponding symbol positions in the original files. 
By contrast to~\cite{Tandon17,CodedPIRPrefetching,CodedPIRPrefetchingJSAC,PIRStorageConstraint20,PIR_SI_20,PIRPatialSI19,Heidarzadeh18,Heidarzadeh19,Heidarzadeh22,LiGastpar20,PIRCodedSI,ChenWangJafar20,Siavoshani21}, the side information in this paper is noisy, for instance, if the files are binary and the test channels are \acp{BSC}, then the client does not know which information bits have been flipped by the \acp{BSC} and which ones have not been flipped.

\textit{Previous works recovered as special cases of our model:} Since the side information considered in this paper is generated by passing the files through some \acp{DMC}, our problem setup can recover the previous works if we assume that the test channels are \acp{BEC}. This is because passing a file through a \ac{BEC} with parameter 0 means that the side information is equal to the input file, and passing a file through a \ac{BEC} with parameter 1 means that there is no side information. For example, when there is only one \ac{BEC} with parameter 1, then the problem studied in this paper subsumes the \ac{PIR} problem~\cite{SunJafar171} and the \ac{PIR} problem with colluding servers~\cite{SunJafar172}. If we assume there are two \acp{BEC} with parameters 0 and 1, then the problem studied in this paper subsumes the \ac{PIR} problem with private noiseless side information~\cite[Theorem~2]{PIR_SI_20}, and the \ac{PIR} problem with colluding servers and private noiseless side information~\cite{ChenWangJafar20} as special cases. If we assume there are $M+1$, where $M\in\bbN_*$, \acp{BEC} with parameters $1-r_i$, then the problem studied here subsumes the \ac{PIR} problem with private side information under storage constraints~\cite{PIRStorageConstraint20}. We provide more details about each of these scenarios after we formally define the problem.
\begin{figure}[t]
\centering
\includegraphics[width=3.0in]{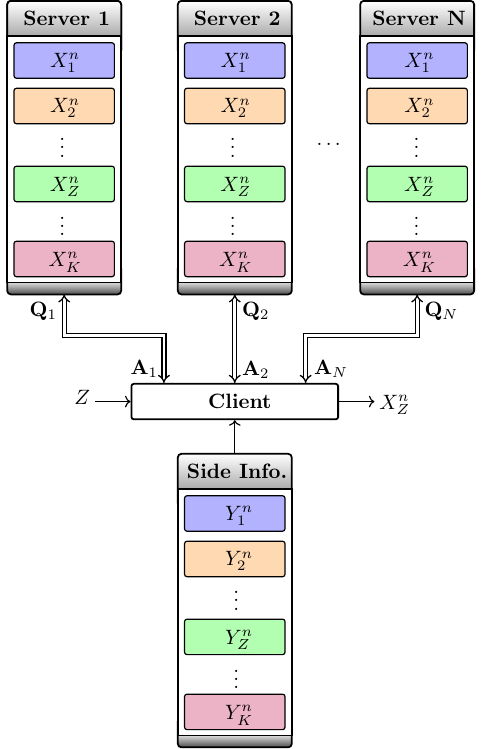}
\caption{\ac{PIR} with private noisy side information and $T$-colluding servers, where the side information about a specific file is obtained by passing this file through one of $D$ possible \acp{DMC} $\left(C^{(i)}\right)_{i\in[D]}$, where $D\le K$, i.e., for $j\in[K]$, there exists some $i\in[D]$ such that $Y_j^n$ is the output of channel $C^{(i)}$ when $X_j^n$ is the input. Here, $\left(X_i^n\right)_{i\in[K]}$ are the $K$ files that are replicated in $N$ servers, $\left(\mathbf{Q}_i\right)_{i\in[N]}$ are the queries for the servers, and $\left(\mathbf{A}_i\right)_{i\in[N]}$ are the corresponding answers of the servers. $Z$ is the index of the client's file selection and $X^n_Z$ is the desired file by the client.}
\label{fig:System_Model}
\vspace{-0.5cm}
\end{figure}
\subsection{Notation}
\label{sec:Notations}
Let $\bbN_*$ be the set of positive natural numbers, and $\bbR$ be the set of real numbers. For any $a, b\in\bbN_*$ such that $a\le b$, $\sbra{a}{b}$ denotes the set $\{a,a+1,\dots,b\}$, $[a]$ denotes the set $\{1,2,\dots,a\}$. Random variables are denoted by capital letters and their realizations by lowercase letters unless specified otherwise. Superscripts denote the dimension of a vector, e.g., $X^n$. For a set of indices $\calI\subset\bbN_*$, $\mathbf{X}_\calI$ denotes $(X_i)_{i\in\calI}$. $\bbE_{X}[\cdot]$ is the expectation with respect to the random variable $X$. The cardinality of a set $\mathcal{S}$ is denoted by $|\mathcal{S}|$. For a mapping $\boldsymbol{\calM}:\calA\to\calB$, the preimage of $b\in\calB$ by $\boldsymbol{\calM}$ is denoted as $\boldsymbol{\calM}^{-1}(b)\triangleq\left\{a\in\calA:\boldsymbol{\calM}(a)=b\right\}$. For $K\in\bbN_*$ and a mapping $\boldsymbol{\calM}:[K]\to\bbR$, we represent the domain and co-domain of $\boldsymbol{\calM}$ as a matrix of dimension $2\times K$ as
\begin{align}
       \boldsymbol{\calM}&=\begin{pmatrix}
1  & 2 & \dots & K \\
\boldsymbol{\calM}(1) & \boldsymbol{\calM}(2) & \dots & \boldsymbol{\calM}(K)
\end{pmatrix}.\nonumber
\end{align}
\subsection{Paper organization}
The remainder of this paper is organized as follows. We formally define the problem in Section~\ref{sec:Prob_Defi}. We present our main results in Section~\ref{sec:Main_Results} and provide the proofs in Section~\ref{sec:Proof_Thm_Single_Server} and Section~\ref{sec:Proof_Capacity_CPIRNSI}. Finally, we provide concluding remarks in Section~\ref{sec:Conclusion}.
\section{Problem Statement}
\label{sec:Prob_Defi}
Consider a client and $N$ servers, where up to $T$ of these $N$ servers may collude, and each server has a copy of $K$ files of length $n$. Additionally, consider a set of $D$ test channels, whose transition probabilities are known to the client and the servers, and whose outputs take value in finite alphabets. We assume that the client has noisy side information about all the  $K$ files in the sense that each file is passed through one of the $D$ test channels, and the output of this test channel is available at the client but not the servers, as depicted in Fig.~\ref{fig:System_Model}. The mapping $\boldsymbol{\calM}$ between the files and the test channels is not known at the servers. The objective of the client is to retrieve one of the files such that the index of this file and the mapping $\boldsymbol{\calM}$ are kept secret from the servers.
\subsection{Problem definitions}
\begin{definition}
\label{defi:Problem_Statement}
Consider $K,N,T,D,n\in\bbN_*$, $\left(d_i\right)_{i\in[D]}\in\bbN_*^D$ such that $\sum_{i=1}^Dd_i=K$, and $D$ distinct test channels $\left(C^{(i)}\right)_{i\in[D]}$, with $C^{(i)}\triangleq (\calX,P_{X|Y}^{(i)},\calY_i)$, where $\calX$ and $\calY_i$, $i\in[K]$, are finite alphabets. Without loss of generality, assume that $H(U|V_i)\le H(U|V_j)$, for $i,j\in[D]$ such that $i\leq j$, where $U$ is uniformly distributed over $\calX$, and $V_i$ and $V_j$ are the outputs of $C^{(i)}$ and $C^{(j)}$, respectively, when $U$ is the input. A \ac{PIR} protocol with private noisy side information and parameters $\left(K,N,T,D,n,\left(d_i\right)_{i\in[D]},\left(C^{(i)}\right)_{i\in[D]}\right)$ consists~of,
\begin{itemize}
    \item $N$ servers, where up to $T$ of these servers may collude;
    \item $K$ independent random sequences $\mathbf{X}_{[K]}^n$ uniformly distributed over $\calX^n$, which represent $K$ files shared at each of the $N$ servers;
    \item $D$ distinct test channels $\left(C^{(i)}\right)_{i\in[D]}$;
    \item a mapping $\boldsymbol{\calM}$ chosen at random from the set $\boldsymbol{\mathfrak{M}}\triangleq\left\{\boldsymbol{\calM}:[K]\to[D]: \forall i\in[D], \abs{\boldsymbol{\calM}^{-1}(i)}=d_i\right\}$; this mapping is only known at the client and not at the servers;
    \item for each file $X_i^n$, where $i\in[K]$, the client has access to a noisy version of $X_i^n$, denoted by $Y_{i,\boldsymbol{\calM}(i)}^n$, which is the output of the test channel $\C^{\left(\boldsymbol{\calM}(i)\right)}$ when $X_i^n$ is the input; 
    \item the random variable $Z$ represents the index of the file that the client wishes to retrieve, i.e., the client wants to retrieve the file~$X_Z^n$; when the client has noiseless side information about some of the files, which means that $C^{(1)}$ is a noiseless test channel, then $Z$ is uniformly distributed over $[K]\backslash\boldsymbol{\calM}^{-1}(1)$, otherwise $Z$ is uniformly distributed over $[K]$;
    \item a stochastic  query function $\calF_i:[K]\times\boldsymbol{\mathfrak{M}}\times\boldsymbol{\calY}_{[K]}^n\to\boldsymbol{\calQ}_i$, for $i\in[N]$, where $\boldsymbol{\calQ}_i$ is a finite alphabet;
    \item for $i\in[N]$, a deterministic answer function $\calE_i:\boldsymbol{\calQ}_i\times\calX^{nK}\to \left[2^{nR\left(\mathbf{Q}_i\right)}\right]$; 
    \item a decoding function $\calD:[K]\times\boldsymbol{\mathfrak{M}}\times\left[2^{n\sum_{i=1}^NR\left(\mathbf{Q}_i\right)}\right]\times\calY^{nK}\to\calX^n$;
\end{itemize}and operates as follows, 
\begin{enumerate}
    \item the client creates the queries $\mathbf{Q}_i\triangleq\calF_i\big(Z,\boldsymbol{\calM},\mathbf{Y}_{[K],\boldsymbol{\calM}}^n\big)$, where $\mathbf{Y}_{[K],\boldsymbol{\calM}}^n\triangleq\left(Y_{i,\boldsymbol{\calM}(i)}^n\right)_{i\in[K]}$, and sends it to Server $i\in[N]$; we assume that the queries must be of negligible length compared to the file length $n$, i.e., $\log\abs{\boldsymbol{\calQ}_i}=o(n)$, for $i\in[N]$;\footnote{When $D=1$ and the test channel is a \ac{BEC} with parameter $\epsilon_1=1$, or when $D=2$ and the test channels are \acp{BEC} with parameters $\epsilon_1=0$ and $\epsilon_2=1$, which correspond to \ac{PIR} without side information in~\cite{SunJafar172} and \ac{PIR} with noiseless side information in~\cite{ChenWangJafar20}, respectively, it is shown in~\cite{SunJafar172,ChenWangJafar20} that there is no loss of generality by making this assumption. In general, allowing the query cost to be non-negligible with the file length $n$ is a different problem. However, similar to~\cite[Remark~1]{PIR_SI_20} and~\cite{PIRStorageConstraint20}, this assumption can also be removed in our converse proofs when the queries $\mathbf{Q}_i$, for $i\in[N]$, are only allowed to depend on $(Z,\boldsymbol{\calM})$.}
    \item then, for all $i\in[N]$, Server~$i$ creates the answer $\mathbf{A}_i\triangleq\calE_i\big(\mathbf{Q}_i,\mathbf{X}_{[K]}^n\big)$, where $\mathbf{X}_{[K]}^n\triangleq\left(X_i^n\right)_{i\in[K]}$, and sends it to the client; therefore,
    \begin{align}
        H\left(\mathbf{A}_i\big|\mathbf{Q}_i,\mathbf{X}_{[K]}^n\right)&=0,\quad\forall i\in[N];\label{eq:Func_Answers}
    \end{align}
    \item finally, the client computes an estimate of $X_Z^n$ as $\calD\left(Z,\boldsymbol{\calM},\mathbf{A}_{[N]},\mathbf{Y}_{[K],\boldsymbol{\calM}}^n\right)$, where $\mathbf{A}_{[N]}\triangleq\left(\mathbf{A}_i\right)_{i\in[N]}$.
\end{enumerate}
Therefore, the probability of error for the client is,
\begin{align}
    &P_{\text{e}}\triangleq\limsup\limits_{n\to\infty}\mathbb{P}\left[\calD\left(Z,\boldsymbol{\calM},\mathbf{A}_{[N]},\mathbf{Y}_{[K],\boldsymbol{\calM}}^n\right)\ne X_Z^n\right].\label{eq:Prob_Error}
\end{align}
$R\left(\mathbf{Q}_{[N]}\right)\triangleq\sum_{i=1}^NR\left(\mathbf{Q}_i\right)$, where $\mathbf{Q}_{[N]}\triangleq\left(\mathbf{Q}_i\right)_{i\in[N]}$, is the normalized download cost of the \ac{PIR} protocol and is random with respect to $\mathbf{Q}_{[N]}$, which makes the protocol a variable length coding scheme. We also define the expected normalized download cost of the protocol as $R\triangleq\bbE_{\mathbf{Q}_{[N]}}[R\left(\mathbf{Q}_{[N]}\right)]$.
\end{definition}
\begin{example}[When $K=D=2$, $T=1$, and $d_1=d_2=1$]
\label{ex:K=D=2_T=1}
Let $X_1^n$ and $X_2^n$ be the two files at the server and $Y_{i,\boldsymbol{\calM}(i)}^n$ be the side information about $X_i^n$, $i\in\{1,2\}$, available at the client but unavailable at the server, where $Y_{i,\boldsymbol{\calM}(i)}^n$ is the output of the test channel $C^{\boldsymbol{(\calM}(i))}$ when the input is $X_i^n$. 
Note that $\boldsymbol{\calM}$ can take two values (with the notation introduced in Section~\ref{sec:Notations}):
\begin{align}
    \mathbf{M}_1&\triangleq\begin{pmatrix}
1  & 2 \\
1 & 2
\end{pmatrix},\quad
    \mathbf{M}_2\triangleq\begin{pmatrix}
1  & 2 \\
2 & 1
\end{pmatrix}.\nonumber
\end{align}
When $Z=1$, since there are two different possibilities for the side information about $X_1^n$, that are $Y_{1,1}^n$ and $Y_{1,2}^n$, we define,
\begin{align}
&P_{\text{e}}^{(1)}\big(Z=1,\mathbf{M}_1\big)\triangleq\nonumber\\
&\mathbb{P}\Big[\calD\left(Z,\boldsymbol{\calM},\mathbf{A}_{[N]},Y_{1,1}^n,Y_{2,2}^n\right)\ne X_1^n\Big|Z=1,\boldsymbol{\calM}=\mathbf{M}_1\Big],\nonumber\\
&P_{\text{e}}^{(2)}\big(Z=1,\mathbf{M}_2\big)\triangleq\nonumber\\
&\mathbb{P}\Big[\calD\left(Z,\boldsymbol{\calM},\mathbf{A}_{[N]},Y_{1,2}^n,Y_{2,1}^n\right)\ne X_1^n\Big|Z=1,\boldsymbol{\calM}=\mathbf{M}_2\Big].\nonumber
\end{align}
Similarly, when $Z=2$, since there are two different possibilities for the side information about $X_2^n$ at the server, that are $Y_{2,1}^n$ and $Y_{2,2}^n$, we define,
\begin{align}
&P_{\text{e}}^{(3)}\big(Z=2,\mathbf{M}_1\big)\triangleq\nonumber\\
&\mathbb{P}\Big[\calD\left(Z,\boldsymbol{\calM},\mathbf{A}_{[N]},Y_{1,1}^n,Y_{2,2}^n\right)\ne X_2^n\Big|Z=2,\boldsymbol{\calM}=\mathbf{M}_1\Big],\nonumber\\
&P_{\text{e}}^{(4)}\big(Z=2,\mathbf{M}_2\big)\nonumber\\
&\triangleq\mathbb{P}\Big[\calD\left(Z,\boldsymbol{\calM},\mathbf{A}_{[N]},Y_{1,2}^n,Y_{2,1}^n\right)\ne X_2^n\Big|Z=2,\boldsymbol{\calM}=\mathbf{M}_2\Big].\nonumber
\end{align}
Therefore, the probability of error in  \eqref{eq:Prob_Error} is equal to,
\begin{align}
&\bbP[Z=1,\boldsymbol{\calM}=\mathbf{M}_1\big]P_{\text{e}}^{(1)}\big(Z=1,\mathbf{M}_1\big)+\nonumber\\
&\bbP[Z=1,\boldsymbol{\calM}=\mathbf{M}_2\big]P_{\text{e}}^{(2)}\left(Z=1,\mathbf{M}_2\right)+\nonumber\\
&\bbP[Z=2,\boldsymbol{\calM}=\mathbf{M}_1\big]P_{\text{e}}^{(3)}\big(Z=2,\mathbf{M}_1\big)+\nonumber\\
&\bbP[Z=2,\boldsymbol{\calM}=\mathbf{M}_2\big]P_{\text{e}}^{(4)}\left(Z=2,\mathbf{M}_2\right).\nonumber
\end{align}
\end{example}
We consider two privacy metrics to study the problem defined above. For the first metric, we keep the index of the desired file $Z$ and the mapping $\boldsymbol{\calM}$ private from the servers, whereas, for the second metric, we allow the index $\boldsymbol{\calM}(Z)$ to be revealed to the server through the queries. As discussed in Section~\ref{sec:ex_defi}, these two privacy metrics recover, as special cases, several \ac{PIR} settings previously studied in the literature.
\begin{definition}[$\C_{\mbox{\scriptsize\rm PIR-PNSI}}$ optimal normalized download cost]
\label{defi:PNSI}
An expected normalized download cost $R\in\bbR_+$ is achievable with private noisy side information and undisclosed side information statistics of the desired file, when up to $T$ servers may collude, if there exist \ac{PIR} protocols such that, for any set $\calT\subseteq[N]$ such that $\abs{\calT}=T$,
\begin{subequations}\label{eq:Achi_Defi_C}
\begin{align}
&P_{\text{e}}= 0,\\
&I\big(\mathbf{Q}_\calT,\mathbf{A}_\calT,\mathbf{X}_{[K]}^n;Z,\boldsymbol{\calM}\big)=0.\label{eq:Sec_Constraint}
\end{align}
\end{subequations}The privacy metric \eqref{eq:Sec_Constraint} means that the client file choice $Z$ and mapping $\boldsymbol{\calM}$ must be kept secret from any $T$ colluding servers. 
The infimum of all achievable normalized download costs is referred to as the \ac{PIR} optimal normalized download cost with private noisy side information and undisclosed side information statistics of the desired file, and is denoted by $\C_{\mbox{\scriptsize\rm PIR-PNSI}}$.
\end{definition}
\begin{definition}[$\C_{\mbox{\scriptsize\rm PIR-PNSI}}^*$ optimal normalized download cost]
\label{defi:CPNSI}
An expected normalized download cost $R\in\bbR_+$ is achievable with private noisy side information and disclosed side information statistics of the desired file, when up to $T$ servers may collude, if there exist \ac{PIR} protocols such that, for any set $\calT\subseteq[N]$ such that $\abs{\calT}=T$,
\begin{subequations}\label{eq:Achi_Defi}
\begin{align}
&P_{\text{e}}= 0,\\
&I\big(\mathbf{Q}_\calT,\mathbf{A}_\calT,\mathbf{X}_{[K]}^n;Z,\boldsymbol{\calM}\big|\boldsymbol{\calM}(Z)\big)=0.\label{eq:Sec_Constraint_C}
\end{align}
\end{subequations}The privacy metric \eqref{eq:Sec_Constraint_C} means that the client file choice $Z$ and the mapping $\boldsymbol{\calM}$ must be kept secret from any $T$ colluding servers, but the noise statistics of the side information of the desired file  available at the client, i.e., $\boldsymbol{\calM}(Z)$, may be revealed to the servers.\footnote{For this privacy metric, we assume that $Z$ is uniformly distributed over $[K]$ because if the client has access to the desired file in a noiseless manner, then according to \eqref{eq:Sec_Constraint_C}, the client can reveal this to the servers and as a result the normalized download cost will be equal to zero.}
This contrasts with privacy metric \eqref{eq:Sec_Constraint} where the noise statistics of the side information of the desired file available at the client must be kept secret from the servers. 
The infimum of all achievable normalized download costs is referred to as the \ac{PIR} optimal normalized download cost with private noisy side information and disclosed side information statistics of the desired file, and is denoted by $\C_{\mbox{\scriptsize\rm PIR-PNSI}}^*$.
\end{definition}
\begin{figure*}
\centering
\includegraphics[width=5.0in]{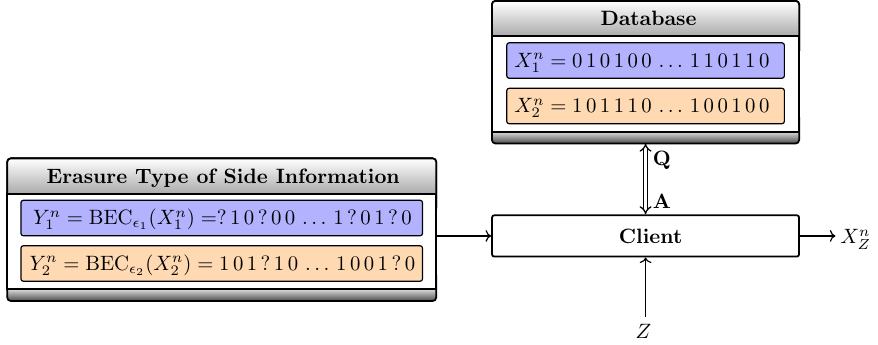}
\vspace{-0.1in}
\caption{Example with $(K,N,T,D)=(2,1,1,2)$ when the test channels are \acp{BEC}.}
\label{fig:Example_Achiev}
\vspace{-0.2in}
\end{figure*}
Note that the privacy constraint in Definition~\ref{defi:PNSI} implies the privacy constraint in Definition~\ref{defi:CPNSI}, i.e., \eqref{eq:Sec_Constraint}$\Rightarrow$\eqref{eq:Sec_Constraint_C}, and we will show that revealing the index of the test channel associated with the desired file $X_Z^n$ can result in a strictly lower normalized download cost. 
When $D=1$, the mapping $\boldsymbol{\calM}$ is deterministic and therefore the secrecy constraints in \eqref{eq:Sec_Constraint} and \eqref{eq:Sec_Constraint_C} are both equal to
\begin{align}
&I\big(\mathbf{Q}_\calT,\mathbf{A}_\calT,\mathbf{X}_{[K]}^n;Z\big)=0.\label{ed:PC_D=1}
\end{align}
\begin{remark}[Rate definition]
Note that in all the previous \ac{PIR} settings, as reviewed in the introduction, the normalized download cost is not a random variable. However, in our setting, since the privacy condition in Definition~\ref{defi:CPNSI} depends on the index of the desired file, the normalized download cost is potentially a random variable. Specifically, we allow the query $\mathbf{Q}_{[N]}$ to depend on $\boldsymbol{\calM}(Z)$, which leads to the normalized download cost $R\left(\mathbf{Q}_{[N]}\right)$ to be a random variable.  
\end{remark}
\subsection{Examples}
\label{sec:ex_defi}
In Example~\ref{ex:PIR_Colluding}, Example~\ref{ex:PIR_Noiseless_SI}, and Example~\ref{ex:PIR_Noiseless_SI_SC}, we show that our problem setup recovers the problem setup for \ac{PIR} with colluding servers~\cite{SunJafar172}, \ac{PIR} with colluding servers and noiseless side information~\cite{PIR_SI_20,ChenWangJafar20}, and \ac{PIR} with private side information under storage constraints~\cite{PIRStorageConstraint20}. 
\begin{example}[\ac{PIR} with colluding servers]
\label{ex:PIR_Colluding}
When $D=1$ and the test channel is a \ac{BEC} with parameter $\epsilon=1$, then the client has no side information about the files. In this case, Definition~\ref{defi:Problem_Statement} reduces to \ac{PIR} without side information as introduced in~\cite{SunJafar172}, and the privacy constraints in Definition~\ref{defi:PNSI} and Definition~\ref{defi:CPNSI} reduce to \eqref{ed:PC_D=1}, which is equivalent to the privacy constraint in~\cite{SunJafar172}.
\end{example}
\begin{example}[\ac{PIR} with private noiseless side information]
\label{ex:PIR_Noiseless_SI}
When $D=2$ and the test channels are \acp{BEC} with parameters $\epsilon_1=0$, and $\epsilon_2=1$, the client has access to $d_1$ files in a noiseless manner as side information. This case corresponds to \ac{PIR} with side information as introduced in~\cite[Theorem~2]{PIR_SI_20} for non-colluding servers and in~\cite[Theorem~1]{ChenWangJafar20} for colluding servers, with the privacy constraint in Definition~\ref{defi:PNSI}.
\end{example}
\begin{example}[\ac{PIR} with private side information under storage constraints]
\label{ex:PIR_Noiseless_SI_SC}
Suppose that $T=1$, $D=M+1$, for $M\in\bbN_*$ and $M\le K$, the test channels are \acp{BEC} with parameters $\epsilon_D=1$, $\epsilon_i=1-r_i$, for $i\in[M]$, with $r_1\ge r_2\ge\dots\ge r_M$, and $d_i=1$, for $i\in[M]$. This problem setup, under the privacy constraint in Definition~\ref{defi:PNSI}, is related to the problem studied in~\cite{PIRStorageConstraint20}. The difference with~\cite{PIRStorageConstraint20} is that the positions of the erasures are known at the servers in \cite{PIRStorageConstraint20}, whereas in our setting, the positions of the erasures are random and unknown at the servers. Therefore, the optimal normalized download cost for our problem setup in this example might be higher than the normalized download cost in~\cite{PIRStorageConstraint20}. However, we will show in the next section that the same normalized download cost as in~\cite{PIRStorageConstraint20} is achievable.
\end{example}

\section{Main Results}
\label{sec:Main_Results}
The novel element of the achievability scheme in this paper is the redundancy removal based on the noisy side information. Before we present our main results, we provide a toy example to illustrate the main ideas of the achievability scheme. 
\subsection{Example with \texorpdfstring{$(K,N,T,D)=(2,1,1,2)$}{(K,N,T,D)=(2,1,1,2)}}In this example, as illustrated in Fig.~\ref{fig:Example_Achiev}, we assume that the files are binary sequences, and the test channels are \acp{BEC} with parameters $\epsilon_1$ and $\epsilon_2$, where $\epsilon_1 < \epsilon_2$. Therefore, the mapping $\boldsymbol{\calM}$ is a random mapping that maps the first file to one of the test channels and maps the other file to the other test channel. Hence, $\boldsymbol{\calM}$ is a random permutation of the set $\{1,2\}$. Here, as seen in Fig.~\ref{fig:Example_SDB}, we first generate the source codes of the files in the database, assuming that the side information of all the files at the client is according to a \ac{BEC} with parameter $\epsilon_1$, by using the Slepian-Wolf encoder \cite{SlepianWolf} \cite[Section~10.4]{ElGamalKim}. We refer to these source codes as $\mathbf{SC}_1$. Similarly, we generate the source codes of the files in the database, assuming that the side information of all the files at the client is according to a \ac{BEC} with parameter $\epsilon_2-\epsilon_1$, by using a second Slepian-Wolf encoder that is nested with the first Slepian-Wolf encoder, and refer to these source codes as $\mathbf{SC}_2$. For the privacy metric defined in Definition~\ref{defi:PNSI}, i.e., when the client does not reveal the index of the test channel associated with the desired file, the client first downloads $\mathbf{SC}_1$, which results in retrieving the file that is associated with the \ac{BEC} with parameter $\epsilon_1$ and also gaining some information about the other file. The normalized download cost for downloading $\mathbf{SC}_1$ is $2\epsilon_1$ \cite{SunJafar171}. According to the Slepian-Wolf Theorem, e.g., \cite{SlepianWolf,ElGamalKim}, the probability of error for retrieving the file that is associated with the \ac{BEC} with parameter $\epsilon_1$ at the client goes to zero as $n$ goes to infinity since the source coding rate is $\epsilon_1$. Next, since the client has noiseless access to the file that is associated with the \ac{BEC} with parameter $\epsilon_1$, and therefore to the source coded version of this file, it then suffices to download $\mathbf{SC}_2(1)\oplus$$\mathbf{SC}_2(2)$, where $\mathbf{SC}_2(i)$, for $i\in\{1,2\}$, is the source code of File $i$, and $\oplus$ denote the modulo $2$ addition. Therefore, the normalized download cost of this operation is $\epsilon_2-\epsilon_1$ \cite[Theorem~2]{PIR_SI_20}, \cite[Therorem~1]{ChenWangJafar20}. The probability of error for decoding the file that is associated with the \ac{BEC} with parameter $\epsilon_2$ at the client goes to zero as $n$ goes to infinity since the source coding rate for this file is $\epsilon_1+(\epsilon_2-\epsilon_1)=\epsilon_2$. The total normalized download cost for this privacy metric is $2\epsilon_1+\epsilon_2-\epsilon_1=\epsilon_1+\epsilon_2$.

Consider now the privacy metric defined in Definition~\ref{defi:CPNSI}, i.e., when the client is willing to reveal the index of the test channel associated with the desired file. When $Z=1$, the client only downloads $\mathbf{SC}_1$, therefore, the normalized download cost is $2\epsilon_1$. When $Z=2$, the client downloads $\mathbf{SC}_1$ and retrieves the file that is associated with the \ac{BEC} with parameter $\epsilon_1$, then it retrieves the desired file by downloading $\mathbf{SC}_2(1)\oplus$$\mathbf{SC}_2(2)$. Therefore, the normalized download cost is $2\epsilon_1+\epsilon_2-\epsilon_1=\epsilon_1+\epsilon_2$. Since $Z$ is a random variable with uniform distribution, the average normalized download cost is $\epsilon_1+\frac{1}{2}(\epsilon_1+\epsilon_2)$. Therefore, the normalized download cost when the client is willing to reveal the index of the test channel associated with the desired file, i.e., Definition~\ref{defi:CPNSI}, is $\frac{1}{2}(\epsilon_2-\epsilon_1)$ less than the normalized download cost when the client does not reveal the index of the test channel associated with the desired file, i.e., Definition~\ref{defi:PNSI}. However, this reduced normalized download cost comes at a price, since the privacy metric in Definition~\ref{defi:PNSI} is stronger than the privacy metric in Definition~\ref{defi:CPNSI}.
\begin{figure*}
\centering
\includegraphics[width=4.0in]{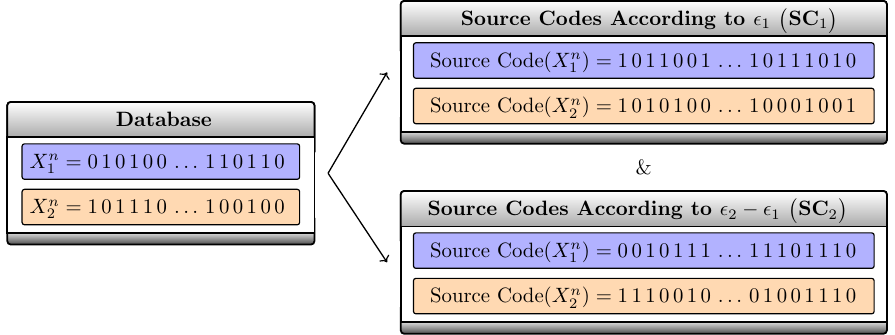}
\vspace{-0.1in}
\caption{Source codes of the files when the side information available at the client is according to the \ac{BEC} with parameter $\epsilon_1$, i.e., $\mathbf{SC}_1$, and source codes of the files when the side information available at the client is according to the \ac{BEC} with parameter $\epsilon_2-\epsilon_1$, i.e., $\mathbf{SC}_2$. Note that the source codes considered are nested.}
\label{fig:Example_SDB}
\vspace{-0.2in}
\end{figure*}
\subsection{Main results}
We now state our main results and present some examples that recover and extend known results. 
\begin{theorem}
\label{thm:Capacity_PIRNSI}
Consider $K$ files that are replicated in $N$ servers, where up to $T$ of them may collude. Then, the optimal normalized download cost of \ac{PIR} with private noisy side information and undisclosed side information statistics of the desired file is
\begin{align}
    \C_{\mbox{\scriptsize\rm PIR-PNSI}}&=\sum\limits_{\ell=1}^{D}H\big(X_1|Y_{1,\ell}\big)\left(\frac{T}{N}\right)^{d_{[\ell+1:D]}}\times\nonumber\\
    &\qquad\left(1+\frac{T}{N}+\left(\frac{T}{N}\right)^2+\dots+\left(\frac{T}{N}\right)^{d_\ell-1}\right)\nonumber\\
    &=\sum\limits_{\ell=1}^{D}H\big(X_1|Y_{1,\ell}\big)\left(\frac{T}{N}\right)^{d_{[\ell+1:D]}}\Psi^{-1}\left(\frac{T}{N},d_\ell\right),\label{eq:Capacity_PIRNSI}
\end{align}where $\Psi^{-1}(A,B)\triangleq\left(1+A+A^2+\dots+A^{B-1}\right)$, and  for $i,j\in\bbN_*$, $d_{[i:j]}\triangleq\sum_{t=i}^jd_t$, when $i\le j$, and $d_{[i:j]}\triangleq0$, when $i>j$.
\end{theorem}
\begin{proof}
The achievability proof is based on a multilevel nested random binning scheme, that allows the client to use the side information efficiently, and the achievability schemes in~\cite{SunJafar172},~\cite{ChenWangJafar20}. 
Without loss of generality, we assume that the channels are ordered according to their noise level, i.e., we assume $H(U|V_i)\le H(U|V_j)$, for $i,j\in[D]$ such that $i\leq j$, where $U$ is uniformly distributed over $\calX$ and $V_i$ and $V_j$ are the outputs of $C^{(i)}$ and $C^{(j)}$, respectively, when $U$ is the input. In our scheme, for each test channel $C^{(\ell)}$, $\ell\in[D]$, the servers store the source coded version of all the files in a new database, denoted by $\mathbf{SC}_\ell$. These databases are obtained by applying a multilevel nested random binning scheme that performs source coding with side information. The client first downloads $\mathbf{SC}_1$, which results in retrieving the $d_1$ files that are associated with the first test channel and obtaining some information about the other files. To download $\mathbf{SC}_2$ the client uses the $d_1$ files that have been retrieved from the previous level as noiseless side information similar to \cite{ChenWangJafar20}. Downloading $\mathbf{SC}_2$ results in retrieving the $d_2$ files that are associated with the second test channel and obtaining some information about the other files. The client continues this process to download all the $(\mathbf{SC}_1,\ldots, \mathbf{SC}_D)$ new databases and, in each step, uses all the files retrieved from the previous steps as noiseless side information. When the client is willing to reveal the index of the test channel that is associated with the desired file, which is denoted by $i$, it only downloads ($\mathbf{SC}_1,\ldots,\mathbf{SC}_i$). Our converse proof shows that this scheme is optimal. 
The details of the achievability proof are available in Section~\ref{sec:Achie_Proof_NNL}. The converse proof is presented in Section~\ref{sec:Converse_Proof}.
\end{proof}
\begin{remark}[Index of random variables]
\label{rem:indices}
Since all the files are generated according to the same distribution, namely, the uniform distribution over $\calX^n$, the index 1 of $X_1$ and $Y_{1,\ell}$ in Theorem~\ref{thm:Capacity_PIRNSI} can be replaced with any other index $i\in[K]$.
\end{remark}
\begin{corollary}
\label{cor:Capacity_PIRNSI}
Consider $K$ files that are replicated in $N$ servers,  where up to $T$ of them may collude. Additionally, the test channels are \acp{BEC} with parameters $\left(\epsilon_i\right)_{i\in[D]}\in[0,1]^D$ such that $\epsilon_i<\epsilon_j$, for $i,j\in\bbN_*$ and $i<j$. Then, the optimal normalized download cost of \ac{PIR} with private noisy side information and undisclosed side information statistics of the desired file is
\begin{align}
    \C_{\mbox{\scriptsize\rm PIR-PNSI}}&=\sum\limits_{\ell=1}^{D}\epsilon_\ell\left(\frac{T}{N}\right)^{d_{[\ell+1:D]}}\times\nonumber\\
    &\qquad\left(1+\frac{T}{N}+\left(\frac{T}{N}\right)^2+\dots+\left(\frac{T}{N}\right)^{d_\ell-1}\right).\nonumber
\end{align}
\end{corollary}
\begin{example}[No side information]
\label{rem:D=1_PIRNSI}
In Corollary~\ref{cor:Capacity_PIRNSI}, if we set $D=1$, and $\epsilon_1=1$, which means that the client has no side information and $d_D=K$, then the optimal normalized download cost result in Corollary~\ref{cor:Capacity_PIRNSI} reduces to~\cite[Theorem~1]{SunJafar172}, i.e.,
\begin{align}
    \C_{\mbox{\scriptsize\rm PIR-PNSI}}&=\left(1+\frac{T}{N}+\left(\frac{T}{N}\right)^2+\dots+\left(\frac{T}{N}\right)^{K-1}\right).\nonumber
\end{align}
\end{example}
\begin{example}[Private noiseless side information]
In Corollary~\ref{cor:Capacity_PIRNSI}, if we set $D=2$, $T=1$, $\epsilon_1=0$, which means that the client knows $d_1$ files as side information in a noiseless manner, and $\epsilon_2=1$, which means that there is no side information about $d_2=K-d_1$ files, then the optimal normalized download cost result in Corollary~\ref{cor:Capacity_PIRNSI} reduces to~\cite[Theorem~2]{PIR_SI_20}, i.e.,
\begin{align}
    \C_{\mbox{\scriptsize\rm PIR-PNSI}}&=\left(1+\frac{1}{N}+\left(\frac{1}{N}\right)^2+\dots+\left(\frac{1}{N}\right)^{K-d_1-1}\right).\nonumber
\end{align}
\end{example}
\begin{example}[Erasure side information with $D=1$]
\label{ex:D=1_BEC}
In Corollary~\ref{cor:Capacity_PIRNSI}, if we set $D=1$ and the test channel to be a \ac{BEC} with parameter $\epsilon$, then the result in Corollary~\ref{cor:Capacity_PIRNSI} reduces to
\begin{align}
    \C_{\mbox{\scriptsize\rm PIR-PNSI}}&=\epsilon\left(1+\frac{T}{N}+\left(\frac{T}{N}\right)^2+\dots+\left(\frac{T}{N}\right)^{K-1}\right).\label{cor:One_Level_Erasure}
\end{align}The optimal result in Example~\ref{ex:D=1_BEC} 
is equal to the optimal normalized download cost of the \ac{PIR} problem when the side information is known by all the terminals.
\end{example}
\begin{example}[\ac{PIR} with private side information under storage constraints]
\label{ex:storage_Cons}
Set $T=1$, $D=M+1$, for $M\in\bbN_*$ and $M< K$. If we set $d_i=1$, for $i\in[M]$, $d_{M+1}=K-M$, and $\epsilon_i=1-r_i$, with $r_1\ge r_2\ge\dots\ge r_M$, and $\epsilon_D=1$, then the optimal normalized download cost in Corollary~\ref{cor:Capacity_PIRNSI} reduces to
\begin{align}
    \C_{\mbox{\scriptsize\rm PIR-PNSI}}&=\frac{1-r_1}{N^{K-1}}+\frac{1-r_2}{N^{K-2}}+\frac{1-r_3}{N^{K-3}}+\dots+\frac{1-r_{M-1}}{N^{K-M+1}}+\nonumber\\
    &\quad\frac{1-r_{M}}{N^{K-M}}+1+\frac{1}{N}+\frac{1}{N^2}+\dots+\frac{1}{N^{K-M-1}}.\nonumber
\end{align}
Note that this result is stronger than that of~\cite[Theorem~1]{PIRStorageConstraint20}, since in~\cite{PIRStorageConstraint20} it is assumed that the client knows the \textit{first} $r_i$ bits, for $i\in[M]$, of $M$ randomly selected files, however, our result in Corollary~\ref{cor:Capacity_PIRNSI} reduces to the same optimal download cost as the optimal download cost derived in \cite{PIRStorageConstraint20} by removing the constraint that the client knows the \textit{first} $r_i$ bits of $M$ files and assumes that the client knows \textit{any randomly chosen} $r_i$ bits of $M$ files.
\end{example}
\begin{subequations}\label{eq:capacity_C}
\begin{theorem}
\label{thm:Capacity_CPIRNSI}
Consider $K$ files, $N$ replicated servers, where up to $T$ of them may collude, and $D$ test channels $\left(C^{(i)}\right)_{i\in[D]}$ as in Definition~\ref{defi:Problem_Statement}. Then, the optimal normalized download cost of \ac{PIR} with private noisy side information and disclosed side information statistics of the desired file is
\begin{align}
    &\C_{\mbox{\scriptsize\rm PIR-PNSI}}^*=\bbE_U[R(U)]\nonumber\\
    &=\frac{1}{K}\sum\limits_{\ell=1}^{D}H\big(X_1|Y_{1,\ell}\big)\left[d_{[\ell+1:D]}\sum\limits_{j=d_{[\ell+1:D]}}^{-1+d_{\ell}+d_{[\ell+1:D]}}\left(\frac{T}{N}\right)^{j} \right.\nonumber\\
    &\qquad\left.+ d_\ell\sum\limits_{j=0}^{-1+d_{[\ell:D]}}\left(\frac{T}{N}\right)^j\right]\nonumber\\
    &=\frac{1}{K}\sum\limits_{\ell=1}^{D}H\big(X_1|Y_{1,\ell}\big)\left[d_{[\ell+1:D]}\left(\frac{T}{N}\right)^{d_{[\ell+1:D]}}\Psi^{-1}\left(\frac{T}{N},d_{\ell}\right)\right.\nonumber\\
    &\qquad\left.+d_\ell \Psi^{-1}\left(\frac{T}{N},d_{[\ell:D]}\right)\right],\label{eq:capacity_Average}
\end{align}where 
\begin{align}
    R(U)&\triangleq\sum\limits_{\ell=1}^{U-1}H\big(X_1|Y_{1,\ell}\big)\left(\frac{T}{N}\right)^{d_{[\ell+1:D]}}\Psi^{-1}\left(\frac{T}{N},d_{\ell}\right)+\nonumber\\
    &\qquad H\big(X_1|Y_{1,U}\big)\Psi^{-1}\left(\frac{T}{N},d_{[U:D]}\right),\label{eq:capacity_NonAverage}
\end{align} with $U$ distributed according to $\bbP[U=u]\triangleq\frac{d_u}{K}$, for $u\in[D]$, $\Psi^{-1}(A,B)\triangleq\left(1+A+A^2+\dots+A^{B-1}\right)$, and  for $i,j\in\bbN_*$, $d_{[i:j]}\triangleq\sum_{t=i}^jd_t$, when $i\le j$, and $d_{[i:j]}\triangleq0$, when $i>j$.
\end{theorem}
\end{subequations}
\begin{proof}
Similar to the achievability scheme of Theorem~\ref{thm:Capacity_PIRNSI}, the achievability scheme of Theorem~\ref{thm:Capacity_CPIRNSI} is based on source coding with side information, and the achievability schemes in~\cite{SunJafar172},~\cite{ChenWangJafar20}. 
{Specifically, we use the same achievability scheme as in Theorem~\ref{thm:Capacity_PIRNSI} by using $\boldsymbol{\calM}(Z)$, instead of $D$, nested random bin indices for each file.}  
The details of the proof are available in Section~\ref{sec:Achie_Proof_C}. The converse proof is presented in Section~\ref{sec:Converse_Proof_CPIRNSI}. 
\end{proof}
The optimal results in Theorem~\ref{thm:Capacity_PIRNSI} and Theorem~\ref{thm:Capacity_CPIRNSI} show that the optimal normalized download cost grows linearly with $H\big(X_1|Y_{1,\ell}\big)$, for $\ell\in[D]$, which quantifies how noisy the side information is. This confirms the intuition that the noisier the side information is, the higher the normalized download cost will become. Note that, the same remark as Remark~\ref{rem:indices} also applies to Theorem~\ref{thm:Capacity_CPIRNSI}.
\begin{corollary}[Binary erasure test channels]
\label{Cor:BEC_Capacity}
Consider $K$ files and $N$ replicated servers, where up to $T$ of them may collude. Additionally, assume that the test channels are \acp{BEC} with parameters $\left(\epsilon_i\right)_{i\in[D]}\in[0,1]^D$ such that $\epsilon_i<\epsilon_j$, for $i,j\in\bbN_*$ and $i<j$. Then, the optimal normalized download cost of \ac{PIR} with private noisy side information and disclosed side information statistics of the desired file~is
\begin{subequations}
\begin{align}
    &\C_{\mbox{\scriptsize\rm PIR-PNSI}}^*=\bbE_U[R(U)]\nonumber\\
    &=\frac{1}{K}\sum\limits_{\ell=1}^{D}\epsilon_\ell\left[d_{[\ell+1:D]}\left(\frac{T}{N}\right)^{d_{[\ell+1:D]}}\Psi^{-1}\left(\frac{T}{N},d_{\ell}\right)+\right.\nonumber\\
    &\qquad\left.d_\ell \Psi^{-1}\left(\frac{T}{N},d_{[\ell:D]}\right)\right],\label{eq:capacity_Average_BEC}\\
  R(U)&\triangleq\sum\limits_{\ell=1}^{U-1}\epsilon_\ell\left(\frac{T}{N}\right)^{d_{[\ell+1:D]}}\Psi^{-1}\left(\frac{T}{N},d_{\ell}\right)+\nonumber\\
    &\qquad\epsilon_U\Psi^{-1}\left(\frac{T}{N},d_{[U:D]}\right).\label{eq:capacity_NonAverag_BEC}
\end{align}
\end{subequations}
\end{corollary}
\begin{figure*}[t!]
\centering
\includegraphics[width=4.0in]{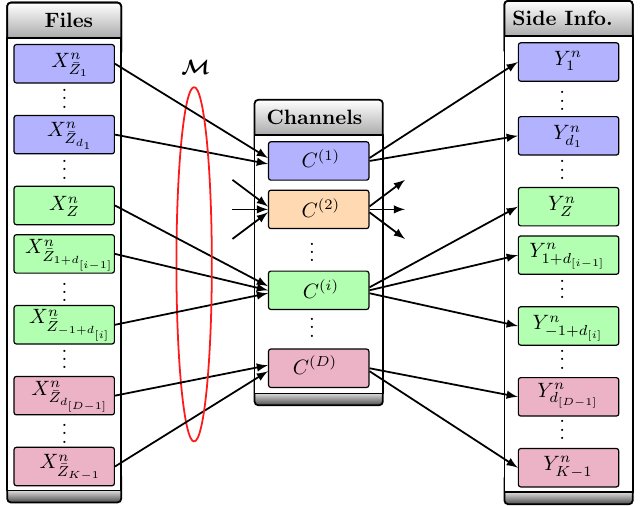}
\caption{Indexing the files based on the mapping $\boldsymbol{\calM}$.}
\label{fig:Mapping}
\vspace{-0.5cm}
\end{figure*}
\begin{figure*}[b!]
\hrulefill
\setcounter{equation}{12}
\begin{align}
    \boldsymbol{\calM}&=\begin{pmatrix}
\mathbf{\bar{Z}}_{[1:d_1]}  & \dots& \left(Z,\mathbf{\bar{Z}}_{[1+d_{[i-1]}:-1+d_{[i]}]}\right)& \dots& \mathbf{\bar{Z}}_{[d_{[D-1]}:-1+d_{[D]}]} \\
(1,\dots,1)  & \dots & \left(i,i,\dots,i\right)& \dots & (D,\dots,D)
\end{pmatrix},\label{eq:Mapping}\\
    \mathbf{Y}_{[K],\boldsymbol{\calM}}^n&\triangleq\left(\mathbf{Y}_{\boldsymbol{\calM}^{-1}(i),i}^n\right)_{i\in[D]}=\left(\mathbf{Y}_{\mathbf{\bar{Z}}_{[1:d_1]},1}^n,\mathbf{Y}_{\mathbf{\bar{Z}}_{\left[1+d_1:d_{[2]}\right]},2}^n,\dots,\mathbf{Y}_{\mathbf{\bar{Z}}_{\left[1+d_{[i-2]}:d_{[i-1]}\right]},i-1}^n,Y_{Z,i}^n,\right.\nonumber\\
    &\quad\left.\mathbf{Y}_{\mathbf{\bar{Z}}_{\left[1+d_{[i-1]}:-1+d_{[i]}\right]},i}^n,\mathbf{Y}_{\mathbf{\bar{Z}}_{\left[d_{[i]}:-1+d_{[i+1]}\right]},i+1}^n,\mathbf{Y}_{\mathbf{\bar{Z}}_{\left[d_{[i+1]}:-1+d_{[i+2]}\right]},i+2}^n,\dots, \mathbf{Y}_{\mathbf{\bar{Z}}_{\left[d_{[D-1]}:-1+d_{[D]}\right]},D}^n\right).\label{eq:Y_Defi}
\end{align}
\setcounter{equation}{9}
\end{figure*}
\begin{corollary}[Binary symmetric channels]
\label{Cor:BSC_Capacity}
When the test channels $C^{(\ell)}$, for $\ell\in[D]$, are \acp{BSC} with parameters $0\le p_1< p_2<\dots< p_D\le\frac{1}{2}$, then the optimal normalized download cost of the \ac{PIR} problem with private noisy side information and disclosed side information statistics of the desired file is 
\begin{align}
    &\C_{\mbox{\scriptsize\rm PIR-PNSI}}^*\nonumber\\
    &=\frac{1}{K}\sum\limits_{\ell=1}^{D}H(p_\ell)\left[d_{[\ell+1:D]}\left(\frac{T}{N}\right)^{d_{[\ell+1:D]}}\Psi^{-1}\left(\frac{T}{N},d_{\ell}\right)+\right.\nonumber\\
    &\qquad\left.d_\ell \Psi^{-1}\left(\frac{T}{N},d_{[\ell:D]}\right)\right],\label{eq:capacity_Average_BSC}
\end{align}where $H(p_\ell)\triangleq -p_\ell\log(p_\ell)-(1-p_\ell)\log(1-p_\ell)$.
\end{corollary}
\begin{example}[Private noiseless side information]
\label{ex:PNLSI}
In Corollary~\ref{Cor:BEC_Capacity}, if we set $D=2$, $\epsilon_1=0$, which means that the client knows $d_1$ files as side information in a noiseless manner, $\epsilon_2=1$, which means that there is no side information about $d_2$ files, and $\boldsymbol{\calM}(Z)=1$, which means that the intended file is included in the noiseless side information, then the normalized download cost $R(U)$ in \eqref{eq:capacity_NonAverag_BEC} is equal to zero. When $\boldsymbol{\calM}(Z)=2$, which means that the intended file is not included in the noiseless side information, then the optimal normalized download cost $R(U)$ in \eqref{eq:capacity_NonAverag_BEC} reduces to, 
\begin{align}
    R(U)&=\left(1+\frac{T}{N}+\left(\frac{T}{N}\right)^2+\dots+\left(\frac{T}{N}\right)^{d_2-1}\right),\nonumber
\end{align}which is the result in~\cite[Theorem~2]{PIR_SI_20}, with $T=1$, and \cite[Theorem~1]{ChenWangJafar20}. Additionally, on average over the file choice, the expected normalized download cost is $\C_{\mbox{\scriptsize\rm PIR-PNSI}}^*=\bbE_U[R(U)]$.
\end{example}
\begin{example}[When $D=1$]
When $D=1$, $U=1$, the first term on the \ac{RHS} of \eqref{eq:capacity_NonAverage} is equal to zero, and the optimal normalized download cost in \eqref{eq:capacity_Average} reduces to Theorem~\ref{thm:Capacity_PIRNSI} when $D=1$, that is,
\begin{align}
    &\C_{\mbox{\scriptsize\rm PIR-PNSI}}^*=\C_{\mbox{\scriptsize\rm PIR-PNSI}}\nonumber\\
    &=H\big(X_1|Y_{1,1}\big)\left(1+\frac{T}{N}+\left(\frac{T}{N}\right)^2+\dots+\left(\frac{T}{N}\right)^{K-1}\right).\nonumber
\end{align}
\end{example}
\begin{remark}[Comparing the results in Theorem~\ref{thm:Capacity_PIRNSI} and Theorem~\ref{thm:Capacity_CPIRNSI}]
\label{EX:Comaring_Main_Results}
We rewrite the result in Theorem~\ref{thm:Capacity_PIRNSI} for any $U\in[D]$ as follows,
\begin{align}
    &\C_{\mbox{\scriptsize\rm PIR-PNSI}}=\sum\limits_{\ell=1}^DH\big(X_1|Y_{1,\ell}\big)\left(\frac{T}{N}\right)^{d_{[\ell+1:D]}}\Psi^{-1}\left(\frac{T}{N},d_\ell\right)\nonumber\\
    &=\sum\limits_{\ell=1}^{U-1}H\big(X_1|Y_{1,\ell}\big)\left(\frac{T}{N}\right)^{d_{[\ell+1:D]}}\Psi^{-1}\left(\frac{T}{N},d_\ell\right)\nonumber\\
    &\quad+\sum\limits_{\ell=U}^DH\big(X_1|Y_{1,\ell}\big)\left(\frac{T}{N}\right)^{d_{[\ell+1:D]}}\Psi^{-1}\left(\frac{T}{N},d_\ell\right)\nonumber\\
    &\mathop=\limits^{(a)}R(U)-H(X_1|Y_{1,U})\Psi^{-1}\left(\frac{T}{N},d_{[U:D]}\right)+\nonumber\\
    &\qquad\sum\limits_{\ell=U}^DH\big(X_1|Y_{1,\ell}\big)\left(\frac{T}{N}\right)^{d_{[\ell+1:D]}}\Psi^{-1}\left(\frac{T}{N},d_\ell\right)\nonumber\\
    &\mathop=\limits^{(b)} R(U)+\sum\limits_{\ell=U}^D\Big(H\big(X_1|Y_{1,\ell}\big)-H\big(X_1|Y_{1,U}\big)\Big)\times\nonumber\\
    &\qquad\left(\frac{T}{N}\right)^{d_{[\ell+1:D]}}\Psi^{-1}\left(\frac{T}{N},d_\ell\right)\nonumber\\
    &\mathop\ge\limits^{(c)}R(U),\label{eq:rem_Comarision}
\end{align}where 
\begin{itemize}
    \item[$(a)$] follows from \eqref{eq:capacity_NonAverage};
    \item[$(b)$] follows by expanding $\Psi^{-1}\left(\frac{T}{N},d_{[U:D]}\right)$;
    \item[$(c)$] follows since $H\big(X_1|Y_{1,U})\le H\big(X_1|Y_{1,\ell})$, for $\ell\in[U:D]$.
    \end{itemize}
Therefore, the optimal normalized download cost in \eqref{eq:capacity_Average}, which is the average of $R(U)$, with respect to $U$, is always smaller than or equal to the optimal normalized download cost in Theorem~\ref{thm:Capacity_PIRNSI}, i.e., $\C_{\mbox{\scriptsize\rm PIR-PNSI}}^*\le\C_{\mbox{\scriptsize\rm PIR-PNSI}}$. This shows that revealing the index of the test channel that is associated with the desired file reduces the normalized download cost. \\ 
Hence, the optimal normalized download cost in Theorem~\ref{thm:Capacity_CPIRNSI} is smaller than the optimal normalized download cost in Theorem~\ref{thm:Capacity_PIRNSI}, and the difference between these two quantities increases as the index $\boldsymbol{\calM}(Z)$ of the test channel that is associated with the desired file decreases.
\end{remark}

\section{Proof of Theorem~\ref{thm:Capacity_PIRNSI}}
\label{sec:Proof_Thm_Single_Server}
\subsection{Converse proof}
\label{sec:Converse_Proof}
\begin{subequations}\label{eq:Mapping_Zbar}
Define $\mathbf{Z}\triangleq(Z,\bar{\mathbf{Z}})$, where $\bar{\mathbf{Z}}\triangleq\left(\bar{Z}_1,\bar{Z}_2,\dots,\bar{Z}_{K-1}\right)$, and
\begin{align}
  \begin{cases}
    \mathbf{\bar{Z}}_{[1+d_{[i-1]}:d_{[i]}]}\triangleq\boldsymbol{\calM}^{-1}(i)& \quad\text{if\quad$i<\boldsymbol{\calM}(Z)$} \\
    \mathbf{\bar{Z}}_{[1+d_{[i-1]}:-1+d_{[i]}]}\triangleq\boldsymbol{\calM}^{-1}(i)\backslash \{Z\}& \quad\text{if\quad$i=\boldsymbol{\calM}(Z)$} \\
    \mathbf{\bar{Z}}_{[d_{[i-1]}:-1+d_{[i]}]}\triangleq\boldsymbol{\calM}^{-1}(i)& \quad\text{if\quad$i>\boldsymbol{\calM}(Z)$} \\
  \end{cases},\label{eq:Z_Bar}
\end{align}where $d_{[i]}\triangleq\sum_{j=1}^id_i$, $\mathbf{\bar{Z}}_{[i:j]}\triangleq(\bar{Z}_i,\bar{Z}_{i+1},\dots,\bar{Z}_j)$, and,  by convention, for $a,b\in\bbN_*$ and $a>b$ define $\mathbf{\bar{Z}}_{\sbra{a}{b}}\triangleq\emptyset$. Then, we index all the files 
as depicted in Fig.~\ref{fig:Mapping} such that the mapping $\boldsymbol{\calM}$ can be described as \eqref{eq:Mapping}, provided at the bottom of the next page, (with the notation introduced in Section~\ref{sec:Notations}). From \eqref{eq:Z_Bar} and \eqref{eq:Mapping}, the side information available at the client is \eqref{eq:Y_Defi}, provided at the bottom of the next page,
in which $\mathbf{Y}_{\mathbf{\bar{Z}}_{[i:j]},\ell}^n\triangleq\left(Y_{\bar{Z}_i,\ell}^n,Y_{\bar{Z}_{i+1},\ell}^n,\dots,Y_{\bar{Z}_j,\ell}^n\right)$.
\end{subequations}
\begin{example}
\label{Ex:Z_Bar}
To illustrate the definition of $\bar{\mathbf{Z}}\triangleq(\bar{Z}_1,\dots,\bar{Z}_{K-1})$ in \eqref{eq:Z_Bar}, consider a setting where $d_i=1$, for $i\in[D]$, which means that $D=K$. In this case, $\sum_{t=1}^{i-1}d_t=i-1$ and $\sum_{t=1}^{i}d_t=i$, therefore, the definition of $\bar{Z}_i$, for $i\in[K-1]$, in \eqref{eq:Z_Bar} reduces to
\setcounter{equation}{14}
\begin{align}
  \begin{cases}
    \bar{Z}_i\triangleq\boldsymbol{\calM}^{-1}(i)& \quad\text{if\quad$i<\boldsymbol{\calM}(Z)$} \\
    \bar{Z}_{i-1}\triangleq\boldsymbol{\calM}^{-1}(i)& \quad\text{if\quad$i>\boldsymbol{\calM}(Z)$} \\
  \end{cases}.\label{eq:Z_Bar_121}
\end{align}For example, let $K=D=3$, $Z=2$, and the side information at the client be $\left(Y_{1,3}^n,Y_{2,1}^n,Y_{3,2}^n\right)$, therefore $\boldsymbol{\calM}(Z)=1$, $\bar{Z}_1\triangleq\boldsymbol{\calM}^{-1}(2)=3$, and $\bar{Z}_2\triangleq\boldsymbol{\calM}^{-1}(3)=1$.
\end{example}
The following equations and lemmas are essential for the converse proof. 
From the dependency graph in Fig.~\ref{fig:Chaining_UC} we have
\begin{align}
    I\left(Z,\mathbf{Q}_{[N]};\mathbf{X}_{[K]}^n\Big|\mathbf{Y}_{[K],\boldsymbol{\calM}}^n,\boldsymbol{\calM}\right)&=0.\label{eq:independ_MZQX}
\end{align}
Considering the probability of error in~\eqref{eq:Prob_Error}, by Fano’s inequality~\cite[Section~2.11]{Cover_Book}, we also have
\begin{align} &\max\limits_{z\in[K]}\,\,\,\max\limits_{\substack{\mathbf{M}\in\mathfrak{M}}}H\big(X_Z^n|\mathbf{Q}_{[N]},\mathbf{A}_{[N]},\mathbf{Y}_{[K],\boldsymbol{\calM}}^n,Z=z,\boldsymbol{\calM}=\mathbf{M}\big) \nonumber\\
&=o(n).\label{eq:Fano_NNL}
\end{align}
\begin{figure}
\centering
\includegraphics[width=3.0in]{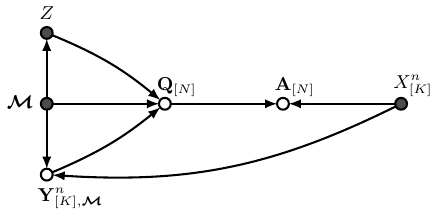}
\vspace{-0.1in}
\caption{Dependency graph for all the involved random variables.}
\label{fig:Chaining_UC}
\vspace{-0.2in}
\end{figure}
\begin{lemma}
\label{lemma:Markov_Chain}
For all $\mathbf{M}\in\boldsymbol{\mathfrak{M}}$, $z\in[K]$, $\calT'\subseteq[N]$, and $\calT\subseteq[N]$ such that $\abs{\calT}=T$, we have 
\begin{align}
    I\left(\mathbf{A}_{\calT};\mathbf{Q}_{[N]\backslash\calT}\Big|\mathbf{Q}_{\calT},\mathbf{Y}_{[K],\boldsymbol{\calM}}^n,\mathbf{X}_{\calT'}^n,Z=z,\boldsymbol{\calM}=\mathbf{M}\right)=0.\label{eq:Markov_Chain_Lemma1}
\end{align}
\end{lemma} 
\begin{proof}
We have,
\begin{align}
    &I\left(\mathbf{A}_{\calT};\mathbf{Q}_{[N]\backslash\calT}\Big|\mathbf{Q}_{\calT},\mathbf{Y}_{[K],\boldsymbol{\calM}}^n,\mathbf{X}_{\calT'}^n,Z=z,\boldsymbol{\calM}=\mathbf{M}\right)\nonumber\\
    &\mathop\leq\limits^{(a)}I\left(\mathbf{A}_{\calT},\mathbf{X}_{[K]}^n;\mathbf{Q}_{[N]\backslash\calT}\Big|\mathbf{Q}_{\calT},\mathbf{Y}_{[K],\boldsymbol{\calM}}^n,\mathbf{X}_{\calT'}^n\right.\nonumber\\
    &\qquad\qquad,Z=z,\boldsymbol{\calM}=\mathbf{M}\Big)\nonumber\\
    &\mathop=\limits^{(b)}I\left(\mathbf{X}_{[K]}^n;\mathbf{Q}_{[N]\backslash\calT}\Big|\mathbf{Q}_{\calT},\mathbf{Y}_{[K],\boldsymbol{\calM}}^n,\mathbf{X}_{\calT'}^n,Z=z,\boldsymbol{\calM}=\mathbf{M}\right)\nonumber\\
    &\quad+I\left(\mathbf{A}_{\calT};\mathbf{Q}_{[N]\backslash\calT}\Big|\mathbf{X}_{[K]}^n,\mathbf{Q}_{\calT},\mathbf{Y}_{[K],\boldsymbol{\calM}}^n,Z=z,\boldsymbol{\calM}=\mathbf{M}\right),\label{eq:Lemma1_Generl_Term}
\end{align}where $(a)$ and $(b)$ hold by the chain rule and non-negativity of the mutual information. The first term on the \ac{RHS} of \eqref{eq:Lemma1_Generl_Term} is equal to zero because of \eqref{eq:independ_MZQX} and the second the term on the \ac{RHS} of \eqref{eq:Lemma1_Generl_Term} is also equal to zero from \eqref{eq:Func_Answers}. 
\end{proof}

\begin{lemma}
\label{lemma:Changing_Z}
For each $\mathbf{M}\in\boldsymbol{\mathfrak{M}}$, $z,z'\in[K]$, $\calT'\subseteq[N]$, and $\calT\subseteq[N]$ such that $\abs{\calT}=T$, we have
\begin{align}
    &H\left(\mathbf{A}_{\calT}\Big|\mathbf{X}_{\calT'}^n,\mathbf{Q}_{\calT},\mathbf{Y}_{[K],\boldsymbol{\calM}}^n,Z=z,\boldsymbol{\calM}=\mathbf{M}\right)\nonumber\\
    &\quad=H\left(\mathbf{A}_{\calT}\Big|\mathbf{X}_{\calT'}^n,\mathbf{Q}_{\calT},\mathbf{Y}_{[K],\boldsymbol{\calM}}^n,Z=z',\boldsymbol{\calM}=\mathbf{M}\right)-o(n).\label{eq:Index_Exchange_Lemma}
\end{align}
\end{lemma}
\begin{proof}
We have,
\begin{align}
    &I\left(\mathbf{Q}_{[N]},\mathbf{A}_{[N]},\mathbf{X}_{[K]}^n,\mathbf{Y}_{[K],\boldsymbol{\calM}}^n;Z\big|\boldsymbol{\calM}=\mathbf{M}\right)\nonumber\\
    &\quad\mathop=\limits^{(a)}I\left(\mathbf{Q}_{[N]},\mathbf{X}_{[K]}^n,\mathbf{Y}_{[K],\boldsymbol{\calM}}^n;Z\big|\boldsymbol{\calM}=\mathbf{M}\right)\nonumber\\
    &\quad=I\left(\mathbf{X}_{[K]}^n,\mathbf{Y}_{[K],\boldsymbol{\calM}}^n;Z\big|\boldsymbol{\calM}=\mathbf{M}\right)+\nonumber\\
    &\quad I\left(\mathbf{Q}_{[N]};Z\big|\mathbf{X}_{[K]}^n,\mathbf{Y}_{[K],\boldsymbol{\calM}}^n,\boldsymbol{\calM}=\mathbf{M}\right)\nonumber\\
    &\quad\mathop\le\limits^{(b)}I\left(\mathbf{X}_{[K]}^n,\mathbf{Y}_{[K],\boldsymbol{\calM}}^n;Z\big|\boldsymbol{\calM}=\mathbf{M}\right)+o(n)\nonumber\\
    &\quad\mathop=\limits^{(c)}I\left(\mathbf{X}_{[K]}^n;Z\big|\mathbf{Y}_{[K],\boldsymbol{\calM}}^n,\boldsymbol{\calM}=\mathbf{M}\right)+o(n)\nonumber\\
    &\quad\mathop=\limits^{(d)}o(n),\label{eq:Two_MI_UC_2}
\end{align}where
\begin{itemize}
    \item[$(a)$] follows from the chain rule and \eqref{eq:Func_Answers};
    \item[$(b)$] follows since the queries are of negligible length compared to the file length $n$ and, therefore, $I\left(\mathbf{Q}_{[N]};Z\big|\mathbf{X}_{[K]}^n,\mathbf{Y}_{[K],\boldsymbol{\calM}}^n,\boldsymbol{\calM}=\mathbf{M}\right)\le H\left(\mathbf{Q}_{[N]}\right)=o(n)$;
    \item[$(c)$] follows from the chain rule and because
   \begin{align}
        &H\left(\mathbf{Y}_{[K],\boldsymbol{\calM}}^n\big|\boldsymbol{\calM}=\mathbf{M},Z\right)\nonumber\\
        &=\sum\limits_{z}\bbP[Z=z\big|\boldsymbol{\calM}=\mathbf{M}]H\left(\mathbf{Y}_{[K],\boldsymbol{\calM}}^n\big|Z=z,\boldsymbol{\calM}=\mathbf{M}\right)\nonumber\\
        &=\sum\limits_{z}\bbP[Z=z\big|\boldsymbol{\calM}=\mathbf{M}]H\left(\mathbf{Y}_{[K],\boldsymbol{\calM}}^n\right)\nonumber\\
        &=H\left(\mathbf{Y}_{[K],\boldsymbol{\calM}}^n\right),\nonumber
    \end{align}where the second equality follows since all the files are generated according to the same distribution and, thus, the entropy of the side information does not depend on $\mathbf{M}$ and $z$.
    \item[$(d)$] follows from \eqref{eq:independ_MZQX}.
    \end{itemize}
    Then, we have
\begin{align}
    o(n)&\mathop=\limits^{(a)}I\left(\mathbf{A}_{\calT},\mathbf{Q}_{\calT},\mathbf{Y}_{[K],\boldsymbol{\calM}}^n,\mathbf{X}_{\calT'}^n;Z\big|\boldsymbol{\calM}=\mathbf{M}\right)\nonumber\\
    &\mathop=\limits^{(b)}I\left(\mathbf{A}_{\calT};Z\big|\mathbf{Q}_{\calT},\mathbf{Y}_{[K],\boldsymbol{\calM}}^n,\mathbf{X}_{\calT'}^n,\boldsymbol{\calM}=\mathbf{M}\right)+o(n),
    \label{eq:Entr_Index_Change_2}
\end{align}where $(a)$ holds by \eqref{eq:Two_MI_UC_2}, and $(b)$ holds by \eqref{eq:Two_MI_UC_2} and the chain rule. 
Finally, \eqref{eq:Entr_Index_Change_2} implies \eqref{eq:Index_Exchange_Lemma}.
\end{proof}
Then, we bound $nR\left(\mathbf{Q}_{[N]}\right)$ as \eqref{eq:RQ_Converse_NNL} provided at the bottom of the next page, where
\begin{itemize}
    \item[$(a)$] follows since conditioning does not increase the entropy;
    \item[$(b)$] follows from Fano's inequality in \eqref{eq:Fano_NNL};
    \item[$(c)$] follows since for $i\triangleq\boldsymbol{\calM}(Z)$, we have
    \begin{align}
    &I\big(\mathbf{Q}_{[N]},\mathbf{Y}_{\bar{\mathbf{Z}},[D]}^n;X_{Z}^n\big|Y_{Z,i}^n,Z,\boldsymbol{\calM}=\mathbf{M}\big)\nonumber\\
    &= I\big(\mathbf{Q}_{[N]};X_Z^n\big|Y_{Z,i}^n,Z,\boldsymbol{\calM}=\mathbf{M}\big)+\nonumber\\
    &\quad I\big(Y_{\bar{\mathbf{Z}},[D]}^n;X_Z^n\big|\mathbf{Q}_{[N]},Y_{Z,i}^n,Z,\boldsymbol{\calM}=\mathbf{M}\big)\nonumber\\
    &=I\big(Y_{\bar{\mathbf{Z}},[D]}^n;X_Z^n\big|\mathbf{Q}_{[N]},Y_{Z,i}^n,Z,\boldsymbol{\calM}=\mathbf{M}\big)\nonumber\\
    &=0,\nonumber
    \end{align}where $\mathbf{Y}_{\bar{\mathbf{Z}},[D]}^n\triangleq\left(\mathbf{Y}_{\bar{Z}_1,[D]}^n,\mathbf{Y}_{\bar{Z}_2,[D]}^n,\dots,\mathbf{Y}_{\bar{Z}_{K-1},[D]}^n\right)$ and  $\mathbf{Y}_{\bar{Z}_i,[D]}^n\triangleq\left(Y_{\bar{Z}_i,1}^n,Y_{\bar{Z}_i,2}^n,\dots,Y_{\bar{Z}_i,D}^n\right)$, for $i\in[K-1]$, and the second equality holds because, from Fig.~\ref{fig:Chaining_UC}, $\mathbf{Q}_{[N]}-(Y_{Z,i}^n,Z,\boldsymbol{\calM})-X_Z^n$ forms a Markov chain, and the last equality holds because, from Fig.~\ref{fig:Chaining_UC}, $Y_{\bar{\mathbf{Z}},[D]}^n-(\mathbf{Q}_{[N]},Y_{Z,i}^n,Z,\boldsymbol{\calM})- X_Z^n$ forms a Markov chain;
    \item[$(d)$] follows since for $i\in[K]$ and $j\in[D]$, $H(X_i^n|Y_{i,j}^n)=H(X_1^n|Y_{1,j}^n)=nH(X_1|Y_{1,j})$ because $P_{X_i^n}=P_{X_1^n}=P_{X}^{\otimes n}$ and, therefore, we have $H\Big(X_z^n\big|Y_{z,j}^n,Z=z,\boldsymbol{\calM}=\mathbf{M}\Big)=H(X_z^n|Y_{z,j}^n)=H(X_1^n|Y_{1,j}^n)=nH(X_1|Y_{1,j})$, for any $z$ and $\mathbf{M}$;
    \item[$(e)$] follows from Lemma~\ref{lemma:Markov_Chain};
    \item[$(f)$] follows from Lemma~\ref{lemma:Changing_Z};
    \item[$(g)$] follows since one can lower bound the second term on the \ac{RHS} of \eqref{eq:d_repeat_UN} using the following inequality,
    \setcounter{equation}{23}
    \begin{subequations}
    \begin{align}
    &H\Big(\mathbf{A}_{[N]}\big|X_{\bar{z}_{K-1}}^n,\mathbf{Q}_{[N]}, \mathbf{Y}_{[K],\boldsymbol{\calM}}^n,Z=\bar{z}_{K-1},\boldsymbol{\calM}=\mathbf{M}\Big)\nonumber\\
    &\ge \frac{1}{\binom{N}{T}}\sum\limits_{\calT:\abs{\calT}=T}H\Big(\mathbf{A}_{\calT}\big|X_{\bar{z}_{K-1}}^n,\mathbf{Q}_{[N]}, \mathbf{Y}_{[K],\boldsymbol{\calM}}^n,\nonumber\\
    &\qquad Z=\bar{z}_{K-2},\boldsymbol{\calM}=\mathbf{M}\Big)\label{eq:repeating_Deviding}\\
    &\ge\frac{T}{N}H\Big(\mathbf{A}_{[N]}\big|X_{\bar{z}_{K-1}}^n,\mathbf{Q}_{[N]}, \mathbf{Y}_{[K],\boldsymbol{\calM}}^n,\nonumber\\
    &\qquad\qquad Z=\bar{z}_{K-2},\boldsymbol{\calM}=\mathbf{M}\Big),\label{eq:Hans}
    \end{align}
    \end{subequations}where \eqref{eq:repeating_Deviding} follows by writing \eqref{eq:One_Subset_T} for all $\binom{N}{T}$ different subsets $\calT\subset[N]$ with cardinality $T$ and adding up all these inequalities; and \eqref{eq:Hans} follows from Han's inequality~\cite[Theorem~17.6.1]{Cover_Book}.
\end{itemize}
\begin{figure*}[b!]
\hrulefill
\setcounter{equation}{22}
\begin{subequations}
    \begin{align}
&nR\left(\mathbf{Q}_{[N]}\right) \nonumber\\
&\ge H\left(\mathbf{A}_{[N]}\right)\nonumber\\
&\mathop\ge\limits^{(a)} H\Big(\mathbf{A}_{[N]}\big|\mathbf{Q}_{[N]}, \mathbf{Y}_{[K],\boldsymbol{\calM}}^n,Z=\bar{z}_{K-1},\boldsymbol{\calM}=\mathbf{M}\Big)\nonumber\\
&= H\Big(\mathbf{A}_{[N]},X_{\bar{z}_{K-1}}^n\big|\mathbf{Q}_{[N]}, \mathbf{Y}_{[K],\boldsymbol{\calM}}^n,Z=\bar{z}_{K-1},\boldsymbol{\calM}=\mathbf{M}\Big)\nonumber\\
&\quad-H\Big(X_{\bar{z}_{K-1}}^n\big|\mathbf{Q}_{[N]},\mathbf{A}_{[N]}, \mathbf{Y}_{[K],\boldsymbol{\calM}}^n,Z=\bar{z}_{K-1},\boldsymbol{\calM}=\mathbf{M}\Big)\label{eq:Repeatiing_UN}\\
&\mathop\ge\limits^{(b)} H\Big(\mathbf{A}_{[N]},X_{\bar{z}_{K-1}}^n\big|\mathbf{Q}_{[N]}, \mathbf{Y}_{[K],\boldsymbol{\calM}}^n,Z=\bar{z}_{K-1},\boldsymbol{\calM}=\mathbf{M}\Big)-o(n)\nonumber\\
&= H\Big(X_{\bar{z}_{K-1}}^n\big|\mathbf{Q}_{[N]}, \mathbf{Y}_{[K],\boldsymbol{\calM}}^n,Z=\bar{z}_{K-1},\boldsymbol{\calM}=\mathbf{M}\Big)\nonumber\\
&\quad+H\Big(\mathbf{A}_{[N]}\big|X_{\bar{z}_{K-1}}^n,\mathbf{Q}_{[N]}, \mathbf{Y}_{[K],\boldsymbol{\calM}}^n,Z=\bar{z}_{K-1},\boldsymbol{\calM}=\mathbf{M}\Big)-o(n)\nonumber\\
&\mathop=\limits^{(c)} H\Big(X_{\bar{z}_{K-1}}^n\big|Y_{\bar{z}_{K-1},D}^n,Z=\bar{z}_{K-1},\boldsymbol{\calM}=\mathbf{M}\Big)\nonumber\\
&\quad+H\Big(\mathbf{A}_{[N]}\big|X_{\bar{z}_{K-1}}^n,\mathbf{Q}_{[N]}, \mathbf{Y}_{[K],\boldsymbol{\calM}}^n,Z=\bar{z}_{K-1},\boldsymbol{\calM}=\mathbf{M}\Big)-o(n)\nonumber\\
&\mathop=\limits^{(d)} nH\Big(X_1\big|Y_{1,D}\Big)+H\Big(\mathbf{A}_{[N]}\big|X_{\bar{z}_{K-1}}^n,\mathbf{Q}_{[N]}, \mathbf{Y}_{[K],\boldsymbol{\calM}}^n,Z=\bar{z}_{K-1},\boldsymbol{\calM}=\mathbf{M}\Big)-o(n)\label{eq:d_repeat_UN}\\
&\ge nH\Big(X_1\big|Y_{1,D}\Big)+H\Big(\mathbf{A}_{\calT}\big|X_{\bar{z}_{K-1}}^n,\mathbf{Q}_{[N]}, \mathbf{Y}_{[K],\boldsymbol{\calM}}^n,Z=\bar{z}_{K-1},\boldsymbol{\calM}=\mathbf{M}\Big)-o(n)\nonumber\\
&\mathop=\limits^{(e)} nH\Big(X_1\big|Y_{1,D}\Big)+H\Big(\mathbf{A}_{\calT}\big|X_{\bar{z}_{K-1}}^n,\mathbf{Q}_{\calT}, \mathbf{Y}_{[K],\boldsymbol{\calM}}^n,Z=\bar{z}_{K-1},\boldsymbol{\calM}=\mathbf{M}\Big)-o(n)\nonumber\\
&\mathop=\limits^{(f)} nH\Big(X_1\big|Y_{1,D}\Big)+H\Big(\mathbf{A}_{\calT}\big|X_{\bar{z}_{K-1}}^n,\mathbf{Q}_{\calT}, \mathbf{Y}_{[K],\boldsymbol{\calM}}^n,Z=\bar{z}_{K-2},\boldsymbol{\calM}=\mathbf{M}\Big)-o(n)\nonumber\\
&\ge nH\Big(X_1\big|Y_{1,D}\Big)+H\Big(\mathbf{A}_{\calT}\big|X_{\bar{z}_{K-1}}^n,\mathbf{Q}_{[N]}, \mathbf{Y}_{[K],\boldsymbol{\calM}}^n,Z=\bar{z}_{K-2},\boldsymbol{\calM}=\mathbf{M}\Big)-o(n)\label{eq:One_Subset_T}\\
&\mathop\ge\limits^{(g)}nH\Big(X_1\big|Y_{1,D}\Big)+\frac{T}{N}H\Big(\mathbf{A}_{[N]}\big|X_{\bar{z}_{K-1}}^n,\mathbf{Q}_{[N]}, \mathbf{Y}_{[K],\boldsymbol{\calM}}^n,Z=\bar{z}_{K-2},\boldsymbol{\calM}=\mathbf{M}\Big)-o(n),
\label{eq:RQ_Converse_NNL}
\end{align}
\end{subequations}
\setcounter{equation}{24}
\end{figure*}
Repeating the steps described in \eqref{eq:RQ_Converse_NNL} starting from \eqref{eq:Repeatiing_UN} with $Z=z'$, where $z'$ changes from the first element till the last element of $\big[\bar{z}_{K-2},\bar{z}_{K-3},\dots,\bar{z}_{d_{[D-1]}}\big]$, to bound the second entropy term on the \ac{RHS} of \eqref{eq:RQ_Converse_NNL}, we obtain \eqref{eq:repetition}, provided at the bottom of the next page, where $(a)$ follows by induction and repeating the steps described in \eqref{eq:RQ_Converse_NNL} starting from \eqref{eq:Repeatiing_UN} with $Z=z'$, where $z'$ changes from the first element till the last element of $\left[\bar{z}_{d_{[D-1]}-1},\bar{z}_{d_{[D-1]}-2},\dots,\bar{z}_{1+d_{[i-1]}},z,\bar{z}_{d_{[i-1]}},\dots,\bar{z}_1\right]$, where $i\triangleq\mathbf{M}(z)$, and $(b)$ follows from~\eqref{eq:Func_Answers}.
\begin{figure*}[b!]
\hrulefill
\setcounter{equation}{24}
\begin{align}
    &R\left(\mathbf{Q}_{[N]}\right) \nonumber \\
    &\ge H\big(X_1|Y_{1,D}\big)+\frac{T}{N}\left[H\big(X_1|Y_{1,D}\big)+\frac{T}{N}\left[H\big(X_1|Y_{1,D}\big)+\dots+\frac{T}{N}\left[H\big(X_1|Y_{1,D}\big)\right.\right.\right.\nonumber\\
    & \left.\left.\left.\qquad+\frac{1}{n}H\Big(\mathbf{A}_{[N]}\big|\mathbf{X}_{\bar{\mathbf{Z}}_{\big[d_{[D-1]}:K-1\big]}}^n,\mathbf{Q}_{[N]}, \mathbf{Y}_{[K],\boldsymbol{\calM}}^n,Z=\bar{z}_{d_{[D-1]}-1},\boldsymbol{\calM}=\mathbf{M}\Big)\nonumber \right]\right]\right]-o(1)\\
    &= 
    H\big(X_1|Y_{1,D}\big)\left[1+\frac{T}{N}+\left(\frac{T}{N}\right)^2+\dots+\left(\frac{T}{N}\right)^{d_D-1}\right]\nonumber\\
    &\qquad+\frac{1}{n}\left(\frac{T}{N}\right)^{d_D}H\Big(\mathbf{A}_{[N]}\big|\mathbf{X}_{\bar{\mathbf{Z}}_{\big[d_{[D-1]}:K-1\big]}}^n,\mathbf{Q}_{[N]}, \mathbf{Y}_{[K],\boldsymbol{\calM}}^n,Z=\bar{z}_{d_{[D-1]}-1},\boldsymbol{\calM}=\mathbf{M}\Big)-o(1)\nonumber\\
    &\mathop\ge\limits^{(a)}\sum\limits_{\ell=1}^{D}H\big(X_1|Y_{1,\ell}\big)\left(\frac{T}{N}\right)^{\sum_{i=\ell+1}^{D}d_i}\left[1+\frac{T}{N}+\left(\frac{T}{N}\right)^2+\dots+\left(\frac{T}{N}\right)^{d_\ell-1}\right]\nonumber\\
    &\qquad+\frac{1}{n}\left(\frac{T}{N}\right)^{K-1}H\Big(\mathbf{A}_{[N]}\big|\mathbf{X}_{[K]}^n,\mathbf{Q}_{[N]}, \mathbf{Y}_{[K],\boldsymbol{\calM}}^n,Z=\bar{z}_1,\boldsymbol{\calM}=\mathbf{M}\Big)-o(1)\nonumber\\
    &\mathop=\limits^{(b)}\sum\limits_{\ell=1}^{D}H\big(X_1|Y_{1,\ell}\big)\left(\frac{T}{N}\right)^{\sum_{i=\ell+1}^{D}d_i}\left[1+\frac{T}{N}+\left(\frac{T}{N}\right)^2+\dots+\left(\frac{T}{N}\right)^{d_\ell-1}\right]-o(1)\label{eq:repetition}
\end{align}
\setcounter{equation}{25}
\end{figure*}

\subsection{Achievability proof}
\label{sec:Achie_Proof_NNL}
A high-level description of the achievability scheme is provided after Theorem~\ref{thm:Capacity_PIRNSI}. 
\subsubsection{Preliminaries}
\label{sec:Achie_Prelim}
Our achievability is based on nested source coding, which we define first and then use in our achievability proof as a black box. Consider a discrete memoryless source $(\mathcal{X}_1 \times \bigtimes_{\ell \in [D]}\calY_{1,\ell},P_{X_1,\mathbf{Y}_{1,[D]}})$ with $D+1$ components. Assume that $(X_1^n,\mathbf{Y}_{1,[D]}^n)$ are \ac{iid} samples of this source. Then, consider an encoder $\calE:\calX_1^n\to\boldsymbol{\calJ}_{[D]}^{(1)}$, that maps the sequence $X_1^n$ to $\mathbf{J}_{[D]}^{(1)}\triangleq (J_{\ell}^{(1)})_{\ell \in [D]}$, where the asymptotic rate of $J_{\ell}^{(1)}$, for $\ell\in[D]$, is $H(X_1|Y_{1,\ell})-H(X_1|Y_{1,\ell-1})$, with the convention $H(X_1|Y_{1,0})=0$. Consider also $D$ decoders $\calD_\ell:\boldsymbol{\calJ}_{[\ell]}^{(1)}\times\calY_{1,\ell}^n\to\calX_1^n$, for $\ell\in[D]$, where the Decoder $\calD_\ell$ assigns an estimate $\hat{X}_1^n$ to $(\mathbf{J}_{[\ell]}^{(1)},Y_{1,\ell}^n)$ such that $\bbP\big[\hat{X}_1^n\ne X_1^n\big]\xrightarrow[n\to\infty]{}0$. In Appendix~\ref{app:Nested_RB}, we explain how to obtain such a scheme with nested random binning and how to implement it with nested polar codes when the side information at the decoders forms a Markov chain.

Assume that each file is of length $n=N^K$, with symbols in a sufficiently large finite field $\bbF_q$. Fix $\delta>0$. 


\subsubsection{Nested Source Coding}
For every file $x_i^n$, $i\in[K]$, generate $D$ nested source codes as in Section~\ref{sec:Achie_Prelim}. For each test channel $\ell\in[D]$, we denote the source code of the file $x_i^n$, $i\in[K]$, by $j_\ell^{(i)}\in\calJ_\ell\eqdef\Big[q_\ell^n\Big]$,  where $q_\ell\triangleq q^{R_\ell}$. Here,  $R_0\triangleq0$ and for $t\in[K]$,
\begin{align}
    \sum_{i=1}^{\ell}R_i= H\big(X_t|Y_{t,\ell}\big)+\delta.\label{eq:SW_Round_ell_N}
\end{align}

We refer to $\mathbf{SC}_\ell\triangleq\Big(j_\ell^{(1)},\dots,j_\ell^{(K)}\Big)_{\ell\in[D]}$ as the database $\ell$. The query is constructed to retrieve each one of the $\mathbf{SC}_\ell$ databases in ascending order.

\subsubsection{Query Structure Construction} 
The client constructs the query in $D$ different levels. In the first level, we apply to the database $\mathbf{SC}_1$  the same query structure as in~\cite{SunJafar172}, which consists of $K$ sublevels. In the level $\ell\in[2:D]$, we apply to the database $\mathbf{SC}_\ell$ the same query structure as in~\cite{ChenWangJafar20}, which also consists of $K$ sublevels. Specifically, as in \cite{ChenWangJafar20}, the $k_\ell$\ts{th} sublevel consists of sums of $k_\ell$ symbols, which are called $k_\ell$-sums. There are $\binom{K}{k_\ell}$ different types of $k_\ell$-sums and $(N-T)^{k_\ell-1}T^{K-k_\ell}$ different instances of each type in the $k_\ell$\ts{th} sublevel. Hence, the total number of symbols that will be downloaded from each server is $\sum_{k_\ell=1}^K\binom{K}{k_\ell}(N-T)^{k_\ell-1}T^{K-k_\ell}$. 

\subsubsection{Query Specialization} For $\ell\in[D]$, we do the query structure construction and query specialization without considering the availability of any side information as in \cite{ChenWangJafar20}, and denote this scheme by $\Pi_\ell$. Then, we do query redundancy removal based on the availability of noiseless side information similar to~\cite{ChenWangJafar20}. Specifically, after each level $\ell\in[D]$, the client is able to recover the $d_{\ell}$ files that are associated with the $\ell$\ts{th} test channel, and therefore considering the files that are decoded in the previous levels, the client knows $\mathbf{X}_{[d_{[\ell]}]}^n$ and, therefore, $\left(j_{\ell+1}^{(1)},\dots,j_{\ell+1}^{\left(d_{[\ell]}\right)}\right)$, which is used as noiseless side information to recover  $\left(j_{\ell+1}^{\left(d_{[\ell]}+1\right)},\dots,j_{\ell+1}^{(K)}\right)$ in level $\ell +1$. For level $\ell=1$, the client does not have any noiseless side information and cannot perform query redundancy removal but, for level $\ell\in\sbra{2}{D}$, since it has recovered $\sum_{t=1}^{\ell-1} d_t$ files, the client can perform query redundancy removal. For each $\ell\in[D]$ and for each server, let $p_{\ell,1}$ denote the number of symbols downloaded with $\Pi_\ell$. Out of these $p_{\ell,1}$ symbols, we denote by $p_{\ell,2}<p_{\ell,1}$ the number of symbols that the client already knows by decoding some of the files in the previous levels. For $\ell\in[D]$, let $\mathbf{U}_{\ell,j}\in\bbF_{q_\ell}^{p_{\ell,1}}$ 
denote the symbols downloaded from the $j$\ts{th} server with $\Pi_\ell$. For each server, use a systematic $(2p_{\ell,1}-p_{\ell,2},p_{\ell,1})$ \ac{MDS} code~\cite{Lin_Costelloo}, with generator matrix $\mathbf{G}_{(2p_{\ell,1}-p_{\ell,2})\times p_{\ell,1}}=[\mathbf{V}_{p_{\ell,1}\times(p_{\ell,1}-p_{\ell,2})}|\mathbf{I}_{p_{\ell,1}\times p_{\ell,1}}]^\intercal$ to encode the $p_{\ell,1}$ symbols into $2p_{\ell,1}-p_{\ell,2}$ symbols, of which $p_{\ell,1}$ are systematic, and $p_{\ell,1}-p_{\ell,2}$ are parity symbols, such that  it is sufficient to download $\mathbf{V}_{p_{\ell,1}\times(p_{\ell,1}-p_{\ell,2})}^\intercal \mathbf{U}_{\ell,j}$. For level  $\ell=1$, since the client does not have any noiseless side information about $\mathbf{SC}_1$, $p_{1,2}=0$.

\subsubsection{Decoding} For $\ell\in[D]$, after reconstructing $(j_i^{(t)})_{i\in[\ell]}$, for $t\in\boldsymbol{\calM}^{-1}(\ell)$, given $\mathbf{Y}_{[K],\boldsymbol{\calM}}^n$, the client forms $\hat{X}_t^n$, an estimate of the sequence $X_t^n$ by using the nested source decoders with \eqref{eq:SW_Round_ell_N}, and thus  
$\bbP\big[\hat{X}_t^n\ne X_t^n\big]\xrightarrow[n\to\infty]{}0$. 
\begin{figure*}[b!]
\hrulefill
\setcounter{equation}{33}
\begin{align}
        &\mathbf{M}_1\triangleq\tau_{z,\bar{z}_{K-1}}\circ\mathbf{M}=\begin{pmatrix}
(\bar{z}_{1} : \bar{z}_{d_1})  & \dots& \Big(z,\bar{z}_{1+d_{[i-1]}},\dots, \bar{z}_{-1+d_{[i]}}\Big)& \dots& \Big(\bar{z}_{d_{[D-1]}} : \bar{z}_{-1+d_{[D]}}\Big) \\
(1,\dots,1)  & \dots & \big(D,i,\dots,i\big)& \dots & \big(D,\dots,D,i\big)
\end{pmatrix},\label{eq:permu_S}
\end{align}
\setcounter{equation}{26}
\end{figure*}
\subsubsection{Rate Calculation}
Similar to~\cite{ChenWangJafar20}, for the scheme $\Pi_\ell$, the total number of downloaded symbols from each server is $p_{\ell,1}=\sum_{k_\ell=1}^K\binom{K}{k_\ell}(N-T)^{k_\ell-1}T^{K-k_\ell}$, $\ell\in[D]$ and out of these $p_{\ell,1}$ symbols $p_{\ell,2}=\sum_{k_\ell=1}^{d_{[\ell-1]}}\binom{d_{[\ell-1]}}{k_\ell}(N-T)^{k_\ell-1}T^{K-k_\ell}$ symbols are already known at the client, where $d_{[\ell-1]}\triangleq\sum_{i=1}^{\ell-1}d_i$ and $d_{[0]}=0$. Then, we have,
\begin{subequations}\label{eq:pis}
\begin{align}
    p_{\ell,1}&=\sum\limits_{k_\ell=1}^K\binom{K}{k_\ell}(N-T)^{k_\ell-1}T^{K-k_\ell}\nonumber\\
    &=\frac{\sum\limits_{k_\ell=0}^K\binom{K}{k_\ell}(N-T)^{k_\ell}T^{K-k_\ell}-T^K}{N-T}\nonumber\\
    &=\frac{N^K-T^K}{N-T},\label{eq:p_1}
\end{align}
similarly, 
\begin{align}
    p_{\ell,2}&=\sum\limits_{k_\ell=1}^{d_{[\ell-1]}}\binom{d_{[\ell-1]}}{k_\ell}(N-T)^{k_\ell-1}T^{K-k_\ell}\nonumber\\
    &=T^{K-d_{[\ell-1]}}\sum\limits_{k_\ell=1}^{d_{[\ell-1]}}\binom{d_{[\ell-1]}}{k_\ell}(N-T)^{k_\ell-1}T^{d_{[\ell-1]}-k_\ell}\nonumber\\
    &=\frac{T^{K-d_{[\ell-1]}}\big(N^{d_{[\ell-1]}}-T^{d_{[\ell-1]}}\big)}{N-T}.\label{eq:p_2}
\end{align} 
\end{subequations}
Therefore, the normalized download cost for the level $\ell$ is,
\begin{align}
    R^{(\ell)}&=\frac{R_\ell N(p_{\ell,1}-p_{\ell,2})}{n}\nonumber\\
    &\mathop=\limits^{(a)}\frac{R_\ell N(p_{\ell,1}-p_{\ell,2})}{N^K}\nonumber\\
    &\mathop=\limits^{(b)}\frac{R_\ell\left(1-\Big(\frac{T}{N}\Big)^{K-d_{[\ell-1]}}\right)}{\big(1-\frac{T}{N}\big)}\nonumber\\
    &\mathop=\limits^{(c)}\Big(H\big(X_1|Y_{1,\ell}\big)-H\big(X_1|Y_{1,\ell-1}\big)\Big)\sum\limits_{i=0}^{K-d_{[\ell-1]}-1}\left(\frac{T}{N}\right)^i,\label{eq:Rl_rate}
\end{align}where
\begin{itemize}
    \item[$(a)$] follows since $n=N^K$;
    \item[$(b)$] follows from \eqref{eq:pis};
    \item[$(c)$] follows from \eqref{eq:SW_Round_ell_N}.
\end{itemize}
Therefore, the total normalized download cost is,
\begin{align}
    \sum\limits_{\ell=1}^DR^{(\ell)}\nonumber&=\sum\limits_{\ell=1}^D\Big(H\big(X_1|Y_{1,\ell}\big)-H\big(X_1|Y_{1,\ell-1}\big)\Big)\times\nonumber\\
    &\qquad\sum\limits_{i=0}^{K-d_{[\ell-1]}-1}\left(\frac{T}{N}\right)^i\nonumber\\
    &=\sum\limits_{\ell=1}^{D}H\big(X_1|Y_{1,\ell}\big)\left(\frac{T}{N}\right)^{K-d_{[\ell]}}\Squad\sum\limits_{i=0}^{d_{\ell}-1}\left(\frac{T}{N}\right)^i.\nonumber
\end{align}

\subsubsection{Privacy Analysis}\label{sec:privacyachievability}
Note that for all the $D$ levels, the client does not use any side information to construct the queries. Indeed, the systematic \ac{MDS} codes of all the levels in the query redundancy removal do not depend on the side information that the client obtains after each level. The decoding starts when the client collects all the answers from the servers for all the $D$ levels. Thus, the side information is used only when the client collects all the answers from the servers for all the $D$ levels. Therefore, privacy is inherited from the privacy of the schemes in~\cite{ChenWangJafar20} and~\cite{SunJafar172}.

\section{Proof of Theorem~\ref{thm:Capacity_CPIRNSI}}
\label{sec:Proof_Capacity_CPIRNSI}
\subsection{Converse proof}
\label{sec:Converse_Proof_CPIRNSI}
The following equations and lemma are essential for the converse proof. Considering the probability of error in \eqref{eq:Prob_Error}, by Fano’s inequality~\cite[Section~2.11]{Cover_Book}, for every $i\in[D]$, we have
\begin{align}
     &\max\limits_{z\in[K]}\,\,\,\max\limits_{\substack{\mathbf{M}\in\mathfrak{M}}}H\big(X_Z^n|\mathbf{Q}_{[N]},\mathbf{A}_{[N]},\mathbf{Y}_{[K],\boldsymbol{\calM}}^n,\boldsymbol{\calM}(Z)=i,\nonumber\\
    &\qquad Z=z,\boldsymbol{\calM}=\mathbf{M}\big)=o(n).\label{eq:Fano_C}
\end{align}
\begin{lemma}
\label{lemma:Changing_Z_C}
For all $z,z'\in[K]$, $i\in[D]$, $\calT,\calT'\subseteq[N]$, and $\mathbf{M},\mathbf{M}'\in\boldsymbol{\mathfrak{M}}$, such that $\mathbf{M}(z) = \mathbf{M}'(z')$, 
\begin{align}
    &H\left(\mathbf{A}_{\calT}\Big|\mathbf{Q}_{\calT},\mathbf{X}_{\calT'}^n,\mathbf{Y}_{[K],\boldsymbol{\calM}}^n,Z=z,\boldsymbol{\calM}(Z)=i,\boldsymbol{\calM}=\mathbf{M}\right)\nonumber\\
    &\quad=H\Big(\mathbf{A}_{\calT}\Big|\mathbf{Q}_{\calT},\mathbf{X}_{\calT'}^n,\mathbf{Y}_{[K],\boldsymbol{\calM}}^n,Z=z',\nonumber\\
    &\qquad\qquad\boldsymbol{\calM}(Z)=i,\boldsymbol{\calM}=\mathbf{M}'\Big)-o(n).\label{eq:Index_Exchange_Lemma_C}
\end{align}
\end{lemma}
\begin{proof}
Since $H(Z)=\log K$ and $H(\boldsymbol{\calM})\le\log\abs{\boldsymbol{\mathfrak{M}}}$, then for all $i\in[D]$, and $\calT,\calT'\subseteq[N]$, we have
\begin{align}
    I\left(\mathbf{A}_\calT;Z,\boldsymbol{\calM}\big|\mathbf{Q}_\calT,\mathbf{X}_{\calT'}^n,\mathbf{Y}_{[K],\boldsymbol{\calM}}^n,\boldsymbol{\calM}(Z)=i\right)=o(n).\label{eq:Two_MI_C_2}
\end{align}
\end{proof}
Consider $(z,\bar{z}_1,\bar{z}_2,\dots,\bar{z}_{K-1})$ a realization of ${\mathbf{Z}}$, and $\mathbf{M}$ a realization of $\boldsymbol{\calM}$ such that $\mathbf{M}(z)=i$. Then, we bound $nR(\mathbf{Q}_{[N]})$ as \eqref{eq:RQ_Converse_NNL_CP}, provided at the bottom of the next page,
\begin{figure*}[b!]
\hrulefill
\setcounter{equation}{31}
\begin{subequations}
\begin{align}
&nR(\mathbf{Q}_{[N]})\nonumber\\&\ge H(\mathbf{A}_{[N]})\nonumber\\
&\mathop\ge\limits^{(a)} H\Big(\mathbf{A}_{[N]}\big|\mathbf{Q}_{[N]},\mathbf{Y}_{[K],\boldsymbol{\calM}}^n,Z=z,\boldsymbol{\calM}(Z)=i,\boldsymbol{\calM}=\mathbf{M}\Big)\nonumber\\
&= H\Big(\mathbf{A}_{[N]},X_z^n\big|\mathbf{Q}_{[N]},\mathbf{Y}_{[K],\boldsymbol{\calM}}^n,Z=z,\boldsymbol{\calM}(Z)=i,\boldsymbol{\calM}=\mathbf{M}\Big)\nonumber\\
&\quad-H\Big(X_z^n\big|\mathbf{Q}_{[N]},\mathbf{A}_{[N]},\mathbf{Y}_{[K],\boldsymbol{\calM}}^n,Z=z,\boldsymbol{\calM}(Z)=i,\boldsymbol{\calM}=\mathbf{M}\Big)\nonumber\\
&\mathop\ge\limits^{(b)} H\Big(\mathbf{A}_{[N]},X_z^n\big|\mathbf{Q}_{[N]},\mathbf{Y}_{[K],\boldsymbol{\calM}}^n,Z=z,\boldsymbol{\calM}(Z)=i,\boldsymbol{\calM}=\mathbf{M}\Big)-o(n)\nonumber\\
&= H\Big(X_z^n\big|\mathbf{Q}_{[N]},\mathbf{Y}_{[K],\boldsymbol{\calM}}^n,Z=z,\boldsymbol{\calM}(Z)=i,\boldsymbol{\calM}=\mathbf{M}\Big)\nonumber\\
&\quad+H\Big(\mathbf{A}_{[N]}\big|X_z^n,\mathbf{Q}_{[N]},\mathbf{Y}_{[K],\boldsymbol{\calM}}^n,Z=z,\boldsymbol{\calM}(Z)=i,\boldsymbol{\calM}=\mathbf{M}\Big)-o(n)\nonumber\\
&\mathop=\limits^{(c)} H\Big(X_z^n\big|Y_{Z,i}^n,Z=z,\boldsymbol{\calM}(Z)=i,\boldsymbol{\calM}=\mathbf{M}\Big)\nonumber\\
&\quad+H\Big(\mathbf{A}_{[N]}\big|X_z^n,\mathbf{Q}_{[N]},\mathbf{Y}_{[K],\boldsymbol{\calM}}^n,Z=z,\boldsymbol{\calM}(Z)=i,\boldsymbol{\calM}=\mathbf{M}\Big)-o(n)\nonumber\\
&\mathop=\limits^{(d)} nH(X_1|Y_{1,i})+H\Big(\mathbf{A}_{[N]}\big|X_z^n,\mathbf{Q}_{[N]},\mathbf{Y}_{[K],\boldsymbol{\calM}}^n,Z=z,\boldsymbol{\calM}(Z)=i,\boldsymbol{\calM}=\mathbf{M}\Big)-o(n)\label{eq:Conditional_step_d}\\
&\ge nH(X_1|Y_{1,i})+H\Big(\mathbf{A}_{\calT}\big|X_z^n,\mathbf{Q}_{[N]},\mathbf{Y}_{[K],\boldsymbol{\calM}}^n,Z=z,\boldsymbol{\calM}(Z)=i,\boldsymbol{\calM}=\mathbf{M}\Big)-o(n)\nonumber\\
&\mathop=\limits^{(e)} nH(X_1|Y_{1,i})+H\Big(\mathbf{A}_{\calT}\big|X_z^n,\mathbf{Q}_{\calT},\mathbf{Y}_{[K],\boldsymbol{\calM}}^n,Z=z,\boldsymbol{\calM}(Z)=i,\boldsymbol{\calM}=\mathbf{M}\Big)-o(n)\nonumber\\
&\mathop=\limits^{(f)} nH(X_1|Y_{1,i})+H\Big(\mathbf{A}_{\calT}\big|X_z^n,\mathbf{Q}_{\calT},\mathbf{Y}_{[K],\boldsymbol{\calM}}^n,Z=\bar{z}_{K-1},\boldsymbol{\calM}(Z)=i,\boldsymbol{\calM}=\mathbf{M}_1\Big)-o(n)\nonumber\\
&\ge nH(X_1|Y_{1,i})+H\Big(\mathbf{A}_{\calT}\big|X_z^n,\mathbf{Q}_{[N]},\mathbf{Y}_{[K],\boldsymbol{\calM}}^n,Z=\bar{z}_{K-1},\boldsymbol{\calM}(Z)=i,\boldsymbol{\calM}=\mathbf{M}_1\Big)-o(n)\label{eq:Conditional_step_g}\\
&\mathop\ge\limits^{(g)}nH(X_1|Y_{1,i})+\frac{T}{N}H\Big(\mathbf{A}_{[N]}\big|X_z^n,\mathbf{Q}_{[N]},\mathbf{Y}_{[K],\boldsymbol{\calM}}^n,Z=\bar{z}_{K-1},\boldsymbol{\calM}(Z)=i,\boldsymbol{\calM}=\mathbf{M}_1\Big)-o(n),
\label{eq:RQ_Converse_NNL_CP}
\end{align}
\end{subequations}
\setcounter{equation}{32}
\end{figure*}
where 
\begin{itemize}
    \item[$(a)$] follows since conditioning does not increase entropy;
    \item[$(b)$] follows from Fano's inequality in \eqref{eq:Fano_C};
    \item[$(c)$] follows since, for $i\triangleq\boldsymbol{\calM}(Z)$, we have,
    \begin{subequations}
    \begin{align}
    &I\big(\mathbf{Q}_{[N]},\mathbf{Y}_{\bar{\mathbf{Z}},[D]}^n;X_{Z}^n\big|Y_{Z,i}^n,Z=z,\boldsymbol{\calM}(Z)=i,\boldsymbol{\calM}=\mathbf{M}\big)\nonumber\\
    &= I\big(\mathbf{Y}_{\bar{\mathbf{Z}},[D]}^n;X_Z^n\big|Y_{Z,i}^n,Z=z,\boldsymbol{\calM}(Z)=i,\boldsymbol{\calM}=\mathbf{M}\big)\nonumber\\
    &\qquad+I\big(\mathbf{Q}_{[N]};X_Z^n\big|\mathbf{Y}_{\bar{\mathbf{Z}},[D]}^n,Y_{Z,i}^n,Z=z,\nonumber\\
    &\qquad\qquad\qquad \boldsymbol{\calM}(Z)=i,\boldsymbol{\calM}=\mathbf{M}\big)\label{eq:Justi_N_1_NNL_1}\\
    &\le o(n),\label{eq:Justi_N_5_NNL_2}
    \end{align}
    \end{subequations}where \eqref{eq:Justi_N_5_NNL_2} holds because the first term on the \ac{RHS} of \eqref{eq:Justi_N_1_NNL_1} is equal to zero since the files are independent of one another, and the second term on the \ac{RHS} of \eqref{eq:Justi_N_1_NNL_1} is less than or equal to $o(n)$ since the queries are of negligible normalized download cost;
    \item[$(d)$] follows because 
    $H\Big(X_Z^n\big|Y_{Z,i}^n,Z=z,\boldsymbol{\calM}(Z)=i,\boldsymbol{\calM}=\mathbf{M}\Big)=H(X_z^n|Y_{z,i}^n)=H(X_1^n|Y_{1,i}^n)=nH(X_1|Y_{1,i})$, for any  $z\in[K]$;
    \item[$(e)$] follows since for all $\mathbf{M}\in\boldsymbol{\mathfrak{M}}$, $z\in[K]$, $\calT'\subseteq[N]$, and $\calT\subseteq[N]$, such that $\abs{\calT}=T$, and $\mathbf{M}(z)=i$ we have
\begin{align}   &I\Big(\mathbf{A}_{\calT};\mathbf{Q}_{[N]\backslash\calT}\Big|\mathbf{Q}_{\calT},\mathbf{Y}_{[K],\boldsymbol{\calM}}^n,\mathbf{X}_{\calT'}^n,Z=z,\boldsymbol{\calM}(Z)=i,\nonumber\\
&\qquad\boldsymbol{\calM}=\mathbf{M}\Big)
\leq H(\mathbf{Q}_{[N]\backslash\calT}) = o(n);\nonumber
\end{align}
    \item[$(f)$] follows from Lemma~\ref{lemma:Changing_Z_C} with $\mathbf{M}_1$ defined as in \eqref{eq:permu_S}, provided at the bottom of this page, where $\tau_{a,b}\circ\mathbf{M}$ is the transposition that exchanges $\mathbf{M}(a)$ and $\mathbf{M}(b)$ in the second row of the matrix $\mathbf{M}$;
    \item[$(g)$] follows since the second term on the \ac{RHS} of \eqref{eq:Conditional_step_d} can be lower bounded by using the following inequality \setcounter{equation}{34}
    \begin{subequations}
   \begin{align}
    &H\Big(\mathbf{A}_{[N]}\big|X_z^n,\mathbf{Q}_{[N]},\mathbf{Y}_{[K],\boldsymbol{\calM}}^n,Z=z,\nonumber\\
    &\qquad\qquad\boldsymbol{\calM}(Z)=i,\boldsymbol{\calM}=\mathbf{M}\Big)\nonumber\\
    &\ge \frac{1}{\binom{N}{T}}\sum\limits_{\calT:\abs{\calT}=T}H\Big(\mathbf{A}_{\calT}\big|X_z^n,\mathbf{Q}_{[N]},\mathbf{Y}_{[K],\boldsymbol{\calM}}^n,\nonumber\\
    &\qquad\qquad Z=\bar{z}_{K-1},\boldsymbol{\calM}(Z)=i,\boldsymbol{\calM}=\mathbf{M}_1\Big)\label{eq:repeating_Deviding_C}\\
    &\ge\frac{T}{N}H\Big(\mathbf{A}_{[N]}\big|X_z^n,\mathbf{Q}_{[N]},\mathbf{Y}_{[K],\boldsymbol{\calM}}^n,\nonumber\\
    &\qquad\qquad Z=\bar{z}_{K-1},\boldsymbol{\calM}(Z)=i,\boldsymbol{\calM}=\mathbf{M}_1\Big),\label{eq:Hans_C}
    \end{align}
    \end{subequations}where \eqref{eq:repeating_Deviding_C} follows by writing \eqref{eq:Conditional_step_g} for all the $\binom{N}{T}$ different subsets $\calT\subseteq[N]$ with cardinality $T$ and adding up all these inequalities; and \eqref{eq:Hans_C} follows from Han's inequality~\cite[Theorem~17.6.1]{Cover_Book}.
\end{itemize}Then, similar to \eqref{eq:RQ_Converse_NNL_CP}, for $\ell\in[K-2]$ with $\mathbf{X}_{\bar{\mathbf{z}}_{[K:K-1]}}^n=\emptyset$, we bound the second term on the \ac{RHS} of \eqref{eq:RQ_Converse_NNL_CP} as provided in \eqref{eq:RQ_Converse_NNL_CP_2},
\begin{figure*}[b!]
\hrulefill
    \setcounter{equation}{35}
\begin{align}
& H\Big(\mathbf{A}_{[N]}\big|\mathbf{X}_{\bar{\mathbf{z}}_{[K-\ell+1:K-1]}}^n,X_z^n,\mathbf{Q}_{[N]},\mathbf{Y}_{[K],\boldsymbol{\calM}}^n,Z=\bar{z}_{K-\ell},\boldsymbol{\calM}(Z)=i,\boldsymbol{\calM}=\mathbf{M}_{\ell}\Big)\nonumber\\
&= H\Big(\mathbf{A}_{[N]},X_{\bar{z}_{K-\ell}}^n\big|\mathbf{X}_{\bar{\mathbf{z}}_{[K-\ell+1:K-1]}}^n,X_z^n,\mathbf{Q}_{[N]},\mathbf{Y}_{[K],\boldsymbol{\calM}}^n,Z=\bar{z}_{K-\ell},\boldsymbol{\calM}(Z)=i,\boldsymbol{\calM}=\mathbf{M}_\ell\Big)\nonumber\\
&\quad-H\Big(X_{\bar{z}_{K-\ell}}^n\big|\mathbf{X}_{\bar{\mathbf{z}}_{[K-\ell+1:K-1]}}^n,X_z^n,\mathbf{Q}_{[N]},\mathbf{A}_{[N]},\mathbf{Y}_{[K],\boldsymbol{\calM}}^n,Z=\bar{z}_{K-\ell},\boldsymbol{\calM}(Z)=i,\boldsymbol{\calM}=\mathbf{M}_\ell\Big)\nonumber\\
&\mathop\ge\limits^{(b)} H\Big(\mathbf{A}_{[N]},X_{\bar{z}_{K-\ell}}^n\big|\mathbf{X}_{\bar{\mathbf{z}}_{[K-\ell+1:K-1]}}^n,X_z^n,\mathbf{Q}_{[N]},\mathbf{Y}_{[K],\boldsymbol{\calM}}^n,Z=\bar{z}_{K-\ell},\boldsymbol{\calM}(Z)=i,\boldsymbol{\calM}=\mathbf{M}_\ell\Big)-o(n)\nonumber\\
&= H\Big(X_{\bar{z}_{K-\ell}}^n\big|\mathbf{X}_{\bar{\mathbf{z}}_{[K-\ell+1:K-1]}}^n,X_z^n,\mathbf{Q}_{[N]},\mathbf{Y}_{[K],\boldsymbol{\calM}}^n,Z=\bar{z}_{K-\ell},\boldsymbol{\calM}(Z)=i,\boldsymbol{\calM}=\mathbf{M}_\ell\Big)\nonumber\\
&\quad+H\Big(\mathbf{A}_{[N]}\big|\mathbf{X}_{\bar{\textbf{z}}_{[K-\ell:K-1}]}^n,X_z^n,\mathbf{Q}_{[N]},\mathbf{Y}_{[K],\boldsymbol{\calM}}^n,Z=\bar{z}_{K-\ell},\boldsymbol{\calM}(Z)=i,\boldsymbol{\calM}=\mathbf{M}_\ell\Big)-o(n)\nonumber\\
&\mathop=\limits^{(c)} H\Big(X_{\bar{z}_{K-\ell}}^n\big|Y_{\bar{z}_{K-\ell},i}^n,Z=\bar{z}_{K-\ell},\boldsymbol{\calM}(Z)=i,\boldsymbol{\calM}=\mathbf{M}_\ell\Big)\nonumber\\
&\quad+H\Big(\mathbf{A}_{[N]}\big|\mathbf{X}_{\bar{\textbf{z}}_{[K-\ell:K-1}]}^n,X_z^n,\mathbf{Q}_{[N]},\mathbf{Y}_{[K],\boldsymbol{\calM}}^n,Z=\bar{z}_{K-\ell},\boldsymbol{\calM}(Z)=i,\boldsymbol{\calM}=\mathbf{M}_\ell\Big)-o(n)\nonumber\\
&\mathop=\limits^{(d)} nH(X_1|Y_{1,i})\nonumber\\
&\quad+H\Big(\mathbf{A}_{[N]}\big|\mathbf{X}_{\bar{\textbf{z}}_{[K-\ell:K-1}]}^n,X_z^n,\mathbf{Q}_{[N]},\mathbf{Y}_{[K],\boldsymbol{\calM}}^n,Z=\bar{z}_{K-\ell},\boldsymbol{\calM}(Z)=i,\boldsymbol{\calM}=\mathbf{M}_\ell\Big)-o(n)\nonumber\\
&\ge nH(X_1|Y_{1,i})\nonumber\\
&\quad+H\Big(\mathbf{A}_{\calT}\big|\mathbf{X}_{\bar{\textbf{z}}_{[K-\ell:K-1}]}^n,X_z^n,\mathbf{Q}_{[N]},\mathbf{Y}_{[K],\boldsymbol{\calM}}^n,Z=\bar{z}_{K-\ell},\boldsymbol{\calM}(Z)=i,\boldsymbol{\calM}=\mathbf{M}_\ell\Big)-o(n)\nonumber\\
&\mathop=\limits^{(e)} nH(X_1|Y_{1,i})\nonumber\\
&\quad+H\Big(\mathbf{A}_{\calT}\big|\mathbf{X}_{\bar{\textbf{z}}_{[K-\ell:K-1}]}^n,X_z^n,\mathbf{Q}_{\calT},\mathbf{Y}_{[K],\boldsymbol{\calM}}^n,Z=\bar{z}_{K-\ell},\boldsymbol{\calM}(Z)=i,\boldsymbol{\calM}=\mathbf{M}_\ell\Big)-o(n)\nonumber\\
&\mathop=\limits^{(f)} nH(X_1|Y_{1,i})\nonumber\\
&\quad+H\Big(\mathbf{A}_{\calT}\big|\mathbf{X}_{\bar{\textbf{z}}_{[K-\ell:K-1}]}^n,X_z^n,\mathbf{Q}_{\calT},\mathbf{Y}_{[K],\boldsymbol{\calM}}^n,Z=\bar{z}_{K-\ell-1},\boldsymbol{\calM}(Z)=i,\boldsymbol{\calM}=\mathbf{M}_{\ell+1}\Big)-o(n)\nonumber\\
&\ge nH(X_1|Y_{1,i})\nonumber\\
&\quad+H\Big(\mathbf{A}_{\calT}\big|\mathbf{X}_{\bar{\textbf{z}}_{[K-\ell:K-1}]}^n,X_z^n,\mathbf{Q}_{[N]},\mathbf{Y}_{[K],\boldsymbol{\calM}}^n,Z=\bar{z}_{K-\ell-1},\boldsymbol{\calM}(Z)=i,\boldsymbol{\calM}=\mathbf{M}_{\ell+1}\Big)-o(n)\nonumber\\
&\mathop\ge\limits^{(g)}nH(X_1|Y_{1,i})\nonumber\\
&\quad+\frac{T}{N}H\Big(\mathbf{A}_{[N]}\big|\mathbf{X}_{\bar{\textbf{z}}_{[K-\ell:K-1}]}^n,X_z^n,\mathbf{Q}_{[N]},\mathbf{Y}_{[K],\boldsymbol{\calM}}^n,Z=\bar{z}_{K-\ell-1},\boldsymbol{\calM}(Z)=i,\boldsymbol{\calM}=\mathbf{M}_{\ell+1}\Big)-o(n),
\label{eq:RQ_Converse_NNL_CP_2}
\end{align}
\setcounter{equation}{36}
\end{figure*}
where $(b)$ to $(g)$ follow with similar arguments of $(b)$ to $(g)$ in \eqref{eq:RQ_Converse_NNL_CP} with $\mathbf{M}_{\ell+1}\triangleq\tau_{\bar{z}_{K-\ell},\bar{z}_{K-\ell-1}}\circ\mathbf{M}_{\ell}$ and $\mathbf{X}_{\bar{\mathbf{z}}_{[K:K-1]}}^n=\emptyset$. The justification of $(c)$ is, however, different. Specifically, for $i\triangleq\boldsymbol{\calM}(Z)$, we have,
\begin{subequations}
    \begin{align}
    &I\big(\mathbf{Q}_{[N]},\mathbf{X}_{[K]\backslash Z}^n,\mathbf{Y}_{\bar{\mathbf{Z}},[D]}^n;X_{Z}^n\big|Y_{Z,i}^n,Z=\bar{z}_{K-\ell},\nonumber\\
    &\qquad\qquad\boldsymbol{\calM}(Z)=i,\boldsymbol{\calM}=\mathbf{M}_\ell\big)\nonumber\\
    &= I\big(\mathbf{X}_{[K]\backslash Z}^n,\mathbf{Y}_{\bar{\mathbf{Z}},[D]}^n;X_Z^n\big|Y_{Z,i}^n,Z=\bar{z}_{K-\ell},\nonumber\\
    &\qquad\qquad\boldsymbol{\calM}(Z)=i,\boldsymbol{\calM}=\mathbf{M}_\ell\big)\nonumber\\
    &\quad+I\big(\mathbf{Q}_{[N]};X_Z^n\big|\mathbf{X}_{[K]\backslash Z}^n,\mathbf{Y}_{\bar{\mathbf{Z}},[D]}^n,Y_{Z,i}^n,Z=\bar{z}_{K-\ell},\nonumber\\
    &\qquad\qquad\boldsymbol{\calM}(Z)=i,\boldsymbol{\calM}=\mathbf{M}_\ell\big)\label{eq:Justi_N_1_NNL_2}\\
    &\le o(n),\label{eq:Justi_N_5_NNL_3}
    \end{align}
    \end{subequations}where \eqref{eq:Justi_N_5_NNL_3} holds because the first term on the \ac{RHS} of \eqref{eq:Justi_N_1_NNL_2} is equal to zero since the files are independent of one another, and the second term on the \ac{RHS} of \eqref{eq:Justi_N_1_NNL_2} is less than or equal to $o(n)$ since the queries are of negligible normalized download cost;

Then, we repeat \eqref{eq:RQ_Converse_NNL_CP_2} to bound the second entropy term on the \ac{RHS} of \eqref{eq:RQ_Converse_NNL_CP_2}, as provided in \eqref{eq:rep_Condi},
\begin{figure*}[b!]
\hrulefill
\setcounter{equation}{37}
\begin{align}
    &R\left(\mathbf{Q}_{[N]}\right)\nonumber\\
    &\mathop\ge\limits^{(a)} H\big(X_1|Y_{1,i}\big)+\frac{T}{N}\left[H\big(X_1|Y_{1,i}\big)+\frac{T}{N}\left[H\big(X_1|Y_{1,i}\big)+\dots+\frac{T}{N}\left[H\big(X_1|Y_{1,i}\big)\right.\right.\right.\nonumber\\
    &\Squad+\left.\left.\left.\frac{1}{n}H\Big(\mathbf{A}_{[N]}\big|\mathbf{X}_{\bar{\mathbf{z}}_{[1+d_{[i-1]}:K-1]}}^n,X_z^n,\mathbf{Q}_{[N]}, \mathbf{Y}_{[K],\boldsymbol{\calM}}^n,Z=\bar{z}_{d_{[i-1]}},\boldsymbol{\calM}=\mathbf{M}_{K-d_{[i-1]}}\Big)\right]\right]\right] - o(1)\nonumber\\
    &= 
    H\big(X_1|Y_{1,i}\big)\left[1+\frac{T}{N}+\left(\frac{T}{N}\right)^2+\dots+\left(\frac{T}{N}\right)^{-1+\sum_{\ell=i}^Dd_i}\right]\nonumber\\
    &\Squad+\frac{1}{n}\left(\frac{T}{N}\right)^{\sum_{\ell=i}^Dd_i}H\Big(\mathbf{A}_{[N]}\big|\mathbf{X}_{\bar{\mathbf{z}}_{[1+d_{[i-1]}:K-1]}}^n,X_z^n,\mathbf{Q}_{[N]}, \mathbf{Y}_{[K],\boldsymbol{\calM}}^n,Z=\bar{z}_{d_{[i-1]}},\boldsymbol{\calM}=\mathbf{M}_{K-d_{[i-1]}}\Big)- o(1)\nonumber\\
    &\mathop\ge\limits^{(b)}H\big(X_1|Y_{1,i}\big)\Psi^{-1}\left(\frac{T}{N},d_{[i:D]}\right)+H\big(X_1|Y_{1,i-1}\big)\left(\frac{T}{N}\right)^{d_{[i:D]}}\Psi^{-1}\left(\frac{T}{N},d_{i-1}\right)\nonumber\\
    &\Squad+\frac{1}{n}\left(\frac{T}{N}\right)^{d_{[i-1:D]}}H\Big(\mathbf{A}_{[N]}\big|\mathbf{X}_{\bar{\mathbf{z}}_{[1+d_{[i-2]}:K-1]}}^n,X_z^n,\mathbf{Q}_{[N]}, \mathbf{Y}_{[K],\boldsymbol{\calM}}^n,Z=\bar{z}_{d_{[i-2]}},\boldsymbol{\calM}=\mathbf{M}_{K-d_{[i-2]}}\Big)- o(1)\nonumber\\
    &\mathop\ge\limits^{(c)}H\big(X_1|Y_{1,i}\big)\Psi^{-1}\left(\frac{T}{N},d_{[i:D]}\right)+\sum\limits_{\ell=1}^{i-1}H\big(X_1|Y_{1,\ell}\big)\left(\frac{T}{N}\right)^{d_{[\ell+1:D]}}\Psi^{-1}\left(\frac{T}{N},d_{\ell}\right)\nonumber\\
    &\Squad+\frac{1}{n}\left(\frac{T}{N}\right)^{K-1}H\Big(\mathbf{A}_{[N]}\big|\mathbf{X}_{[K]}^n,\mathbf{Q}_{[N]}, \mathbf{Y}_{[K],\boldsymbol{\calM}}^n,Z=\bar{z}_1,\boldsymbol{\calM}=\mathbf{M}_{K-1}\Big)- o(1)\nonumber\\
    &\mathop=\limits^{(d)}H\big(X_1|Y_{1,i}\big)\Psi^{-1}\left(\frac{T}{N},d_{[i:D]}\right)+\sum\limits_{\ell=1}^{i-1}H\big(X_1|Y_{1,\ell}\big)\left(\frac{T}{N}\right)^{d_{[\ell+1:D]}}\Psi^{-1}\left(\frac{T}{N},d_{\ell}\right)- o(1),\label{eq:rep_Condi}
\end{align}
\setcounter{equation}{38}
\end{figure*}
where 
\begin{itemize}
    \item[$(a)$] follows by repeating \eqref{eq:RQ_Converse_NNL_CP_2} with $Z=z'$, where $z'$ changes from the first element till the last element of  $\big[\bar{z}_{K-1},\bar{z}_{K-2},\dots,\bar{z}_{1+d_{[i-1]}}\big]$;
    \item[$(b)$] follows by repeating \eqref{eq:RQ_Converse_NNL_CP_2} with $Z=z'$, where $z'$ changes from the first element till the last element of  $\big[\bar{z}_{d_{[i-1]}},\bar{z}_{d_{[i-1]}-1},\dots,\bar{z}_{1+d_{[i-2]}}\big]$;
    \item[$(c)$] follows by induction and repeating $(b)$;
    \item[$(d)$] follows from~\eqref{eq:Func_Answers}.
\end{itemize}
\subsection{Achievability proof}
\label{sec:Achie_Proof_C}
We use the same coding scheme as that of Theorem~\ref{thm:Capacity_PIRNSI} in Section~\ref{sec:Achie_Proof_NNL} with $U\triangleq\boldsymbol{\calM}(Z)$ levels instead of $D$ levels. Therefore, from \eqref{eq:Rl_rate} the total normalized download cost is,
\begin{align}
    R(U)&=\sum\limits_{\ell=1}^UR^{(\ell)}\nonumber\\
    &=\sum\limits_{\ell=1}^U\Big(H\big(X_1|Y_{1,\ell}\big)-H\big(X_1|Y_{1,\ell-1}\big)\Big)\times\nonumber\\
    &\qquad\left(1+\frac{T}{N}+\dots+\left(\frac{T}{N}\right)^{K-1-\sum_{i=1}^{\ell-1}d_i}\right)\nonumber\\
    &=\sum\limits_{\ell=1}^{U-1}H\big(X_1|Y_{1,\ell}\big)\left(\frac{T}{N}\right)^{K-\sum_{i=1}^{\ell}d_i}\times\nonumber\\
    &\qquad\left(1+\frac{T}{N}+\dots+\left(\frac{T}{N}\right)^{d_{\ell}-1}\right)+H\big(X_1|Y_{1,U}\big)\times\nonumber\\
    &\qquad\left(1+\frac{T}{N}+\dots+\left(\frac{T}{N}\right)^{K-1-\sum_{i=1}^{U-1}d_i}\right).\label{eq:Achievable_Rate_Unxpected_NNL}
\end{align}
Calculating the expectation of \eqref{eq:Achievable_Rate_Unxpected_NNL} with respect to $U$ results to
\begin{align}
    &\bbE_U\left[R(U)\right]\nonumber\\
    &=\frac{1}{K}\sum\limits_{u=1}^Dd_u\left[\sum\limits_{\ell=1}^{u-1}H\big(X_1|Y_{1,\ell}\big)\left(\frac{T}{N}\right)^{K-\sum_{i=1}^{\ell}d_i}\times\right.\nonumber\\
    &\qquad\qquad\left(1+\frac{T}{N}+\dots+\left(\frac{T}{N}\right)^{d_{\ell}-1}\right)\nonumber\\
    &\quad+\left.H\big(X_1|Y_{1,u}\big)\left(1+\frac{T}{N}+\dots+\left(\frac{T}{N}\right)^{K-1-\sum_{i=1}^{u-1}d_i}\right)\right]\nonumber\\
    &=\frac{1}{K}\left[\sum\limits_{u=1}^Dd_u\sum\limits_{\ell=1}^{u-1}H\big(X_1|Y_{1,\ell}\big)\left(\frac{T}{N}\right)^{K-\sum_{i=1}^{\ell}d_i}\times\right.\nonumber\\
    &\qquad\left(1+\frac{T}{N}+\dots+\left(\frac{T}{N}\right)^{d_{\ell}-1}\right)+\sum\limits_{u=1}^Dd_uH\big(X_1|Y_{1,u}\big)\times\nonumber\\
    &\quad\left.\left(1+\frac{T}{N}+\dots+\left(\frac{T}{N}\right)^{K-1-\sum_{i=1}^{u-1}d_i}\right)\right]\nonumber\\
    &=\frac{1}{K}\sum\limits_{\ell=1}^{D}H\big(X_1|Y_{1,\ell}\big)\left[\left(K-\sum\limits_{i=1}^{\ell}d_i\right)\left(\frac{T}{N}\right)^{K-\sum_{i=1}^{\ell}d_i}\times\right.\nonumber\\
    &\qquad\left(1+\frac{T}{N}+\dots+\left(\frac{T}{N}\right)^{d_{\ell}-1}\right)\nonumber\\
    &\qquad\left.+d_\ell \left(1+\frac{T}{N}+\dots+\left(\frac{T}{N}\right)^{K-1-\sum_{i=1}^{\ell-1}d_i}\right)\right]\nonumber.
\end{align}
Finally, similar to Section~\ref{sec:privacyachievability}, privacy is inherited from the privacy of the scheme in~\cite{ChenWangJafar20} and~\cite{SunJafar172}.

\section{Conclusion}
\label{sec:Conclusion}
We have studied the \ac{PIR} problem with $N$ servers, where each server has a copy of $K$ files and $T$ of the servers may collude, when the client has a noisy version of each of the $K$ files. The side information is such that each file is passed through one of $D$ possible and distinct test channels, whose statistics are known by the client and the servers. We studied this problem under two different security metrics. Under the first metric,  the client wants to keep the index of the desired file and the mapping between the files and the test channels secret from the servers. Under the second metric, the client wants to keep the index of the desired file and the mapping between the files and the test channels secret from the servers, but is willing to reveal the index of the test channel that is associated with the desired file. We derived the optimal normalized download cost under both privacy metrics. We showed that the optimal normalized download cost under the second privacy metric is smaller than or equal to the optimal normalized download cost under the first privacy metric, which shows that revealing the index of the test channel that is associated with the desired file results in a lower normalized download cost. Our setting and results recover several known settings, including \ac{PIR} with private noiseless side information and \ac{PIR} with private side information under storage constraints. We note that the \ac{PIR} problem with noisy side information when the side information is not required to be kept private is an interesting open problem.

\begin{appendices}
\section{Nested Source Coding Schemes}
\label{app:Nested_RB}
In Appendix~\ref{sec:Nested_RB}, we provide a nested random binning scheme to implement nested source coding, as described in Section~\ref{sec:Achie_Prelim}. In Appendix~\ref{sec:Nested_PC}, we provide an implementation of this scheme with polar codes for the case that the side information forms the Markov chain $X_t-Y_{t,1}-Y_{t,2}-\dots-Y_{t,D}$, for $t\in[K]$. As formalized next, this Markov chain is always satisfied for test channels that are degraded with respect to one another as, for instance, in Corollaries \ref{Cor:BEC_Capacity}, \ref{Cor:BSC_Capacity} for binary erasure or binary symmetric test channels.
\begin{lemma}
If there exists a test channel $\tilde{C}_{i,j}$ such that $C^{(j)}=\tilde{C}_{i,j}\circ C^{(i)}$, for $i,j\in[D]$ and $i<j$, i.e., $C^{(j)}$ is degraded with respect to $C^{(i)}$, then, without loss of generality, one can assume that $X_t-Y_{t,1}-Y_{t,2}-\dots-Y_{t,D}$, for $t\in[K]$, forms a Markov chain.
\end{lemma}
\begin{proof}
Since the test channels are degraded with respect to one another, one can redefine $X_t$, $(Y_{t,i})_{i\in[D]}$, $t\in[K]$, such that $X_t-Y_{t,1}-Y_{t,2}-\dots-Y_{t,D}$ forms a Markov chain. Note that the probability of error in \eqref{eq:Prob_Error} and the privacy condition in  \eqref{eq:Sec_Constraint_C} will not be affected because they do not depend on the joint distribution between $X_t$ and $(Y_{t,i},Y_{t,j})$, for $t\in[K]$, $i,j\in[D]$, and $i\ne j$.
\end{proof}

\subsection{Nested Random Binning Scheme}
\label{sec:Nested_RB}
Consider a discrete memoryless source $(\mathcal{X}_1 \times \bigtimes_{\ell \in [D]}\calY_{1,\ell},P_{X_1,\mathbf{Y}_{1,[D]}})$ with $D+1$ components. Assume that $(X_1^n,\mathbf{Y}_{1,[D]}^n)$ are \ac{iid} samples of this source. Then, consider an encoder $\calE:\calX_1^n\to\boldsymbol{\calJ}_{[D]}^{(1)}$, that assigns $D$ random bin indices $\mathbf{J}_{[D]}^{(1)}\triangleq (J_{\ell}^{(1)})_{\ell \in [D]}$  to the sequence $X_1^n$, where the asymptotic rate of $J_{\ell}^{(1)}$, for $\ell\in[D]$, is $H(X_1|Y_{1,\ell})-H(X_1|Y_{1,\ell-1})$, with the convention $H(X_1|Y_{1,0})=0$. 
Consider $D$ decoders $\calD_\ell:\boldsymbol{\calJ}_{[\ell]}^{(1)}\times\calY_{1,\ell}^n\to\calX_1^n$, for $\ell\in[D]$, such that the decoder $\calD_\ell$ assigns an estimate $\hat{X}_1^n$ to $(\mathbf{J}_{[\ell]}^{(1)},Y_{1,\ell}^n)$ if there is a unique $\hat{X}_1^n$ such that $\left(\hat{X}_1^n,Y_\ell^n\right)$ are jointly typical and $\left(J_1^{(1)},J_2^{(1)},\dots,J_\ell^{(1)}\right)$ corresponds to the first $\ell$ components of $\calE(\hat{X}_1^n)$. According to \cite{SlepianWolf},  \cite[Section~10.4]{ElGamalKim}, since the asymptotic sum rate for the bin indices that are used at the decoder $\calD_\ell$, i.e., $\mathbf{J}_{[\ell]}^{(1)}$, is $H(X_1|Y_{1,\ell})$, then  $\bbP\big[\hat{X}_1^n\ne X_1^n\big]\xrightarrow[n\to\infty]{}0$.

\subsection{Nested Polar Coding Scheme}
\label{sec:Nested_PC}
We now provide an implementation for the nested source coding in Section~\ref{sec:Achie_Proof_NNL} by using nested polar codes when the side information available at the decoders forms a Markov chain. We will rely on the following result for source coding with side information from \cite{arikan2010source}.
\begin{lemma}\label{lemma:lem1p}(Source Coding with side information\cite{arikan2010source}). Consider a probability distribution $p_{XY}$ over $\calX\times\calY$ with $|\calX|=2$ and $\calY$ a finite alphabet. Let $N$ be a power of $2$ and consider $(X^{N}, Y^{N} )$ distributed according to $\prod_{i=1}^N p_{XY}$. Define $U^{N}\triangleq X^{N}G_N$, where  
$G_N \triangleq  \left[ \begin{smallmatrix}
       1 & 0            \\[0.3em]
       1 & 1 
\end{smallmatrix} \right]^{\otimes \log N} $ is the source polarization matrix defined in \cite{arikan2010source}. Define also for $\delta_N\triangleq 2^{-N^\beta}$ with $\beta \in]0,\frac{1}{2}[$, the set $\calH_{X|Y} \triangleq \left\{ i \in [N] 
: H( U_i | U^{i-1}Y^{N}) > \delta_N \right\}$. Given $U^{N}[ \calH_{X|Y}]$ and $Y^{N}$, one can form $\hat U^{N}$ by the successive cancellation decoder of \cite{arikan2010source} such that $\mathbb{P}[\hat U^{N}\neq U^{N}]\leq N\delta_N$. Moreover, $\displaystyle \lim_{N \to \infty}|\calH_{X|Y}|/N=H(X|Y)$. 
\end{lemma}

Let $N = 2^n$. Fix a joint probability distribution       $P_{X_{1}\mathbf{Y}_{1,[D]}}\triangleq P_{X_1} p_{Y_{1,1}|X_1} \prod_{\ell=2}^D P_{Y_{1,\ell}|Y_{1,\ell-1}}$ over $\calX_1\times\boldsymbol{\calY}_{1,[D]}$, where $|\calX_1|=2$, $(\calY_{1,\ell})_{\ell\in[D]}$ are finite alphabets, $\boldsymbol{\calY}_{1,[D]} \triangleq \bigtimes_{\ell\in [D]} \calY_{1,\ell}$, and $\mathbf{Y}_{1,[D]} \triangleq (Y_{1,\ell})_{\ell\in[D]}$.       Define $U_{}^{N} \triangleq X_{}^{N} G_N$. For  $\delta_N \triangleq 2^{-N^{\beta}}$, $\beta \in]0,\frac{1}{2}[$, define for $\ell\in[D]$, 
\begin{align}
\calH_{X|Y_\ell} &\triangleq \left\{ i \in [N]: H( U_i | U^{i-1}Y_{1,\ell}^{N}) \geq \delta_N \right\}.\nonumber
\end{align}

\begin{lemma}
\label{lemma:lem2p}
For $\ell\in [D-1]$, we have $\calH_{X_1|Y_{1,\ell}} \subset \calH_{X_1|Y_{1,\ell+1}}$.
\end{lemma}
\begin{proof}
Let  $i \in \calH_{X_1|Y_{1,\ell}}$. We have
\begin{align*}
    \delta_N 
    &\stackrel{(a)}\leq H( U_i | U^{i-1}Y_{1,\ell}^{N})\\
    & \stackrel{(b)}= H( U_i | U^{i-1}Y_{1,\ell}^{N}Y_{\ell+1}^{N}) \\
    &\stackrel{(c)}\leq H( U_i | U^{i-1}Y_{1,\ell+1}^{N}),
\end{align*}
where $(a)$ holds because $i \in \calH_{X_1|Y_{1,\ell}}$, $(b)$ holds because $I(U^i;Y_{1,\ell+1}^{N}|U^{i-1}Y_{1,\ell}^{N}) \leq I(U^{N};Y_{1,\ell+1}^{N}|Y_{1,\ell}^{N})=0$, $(c)$ holds because conditioning does not increase entropy.
\end{proof}
From Lemmas~\ref{lemma:lem1p} and \ref{lemma:lem2p}, we deduce the following proposition.
\begin{proposition}
Let $\ell\in [D-1]$. Define $J_\ell \triangleq U_{}^{N}[\calH_{X_1|Y_{1,\ell}}]$ and $J'_{\ell+1} \triangleq U_{}^{N}[\calH_{X_1|Y_{1,\ell+1}} \backslash \calH_{X_1|Y_{1,\ell}}]$.
Then, $\displaystyle \lim_{N \to \infty}|J_\ell|/N=H(X_1|Y_{1,\ell})$, $\displaystyle \lim_{N \to \infty}|J'_{\ell+1}|/N=H(X_1|Y_{1,\ell+1})-H(X_1|Y_{1,\ell})$, and one can reconstruct $X^N$ from $(J_\ell,J'_{\ell+1},Y^N_{1,\ell+1})$ with vanishing probability of error as $N$ goes to infinity.
\end{proposition}
\begin{proof}
   We have $|J_\ell|/N=|\calH_{X_1|Y_{1,\ell}}|/N \xrightarrow[N\to\infty]{} H(X_1|Y_{1,\ell})$, where the limit holds by \cite{arikan2010source}. Then, by Lemma~\ref{lemma:lem2p}, we have $|J'_{\ell+1}|/N=|\calH_{X_1|Y_{1,\ell+1}} \backslash \calH_{X_1|Y_{1,\ell}}|/N = |\calH_{X_1|Y_{1,\ell+1}}|/N- | \calH_{X_1|Y_{1,\ell}}|/N \xrightarrow[N\to\infty]{} H(X_1|Y_{1,\ell+1})-H(X_1|Y_{1,\ell})$, where the limit holds by \cite{arikan2010source}. Finally, the near lossless reconstruction of $X_1^N$ from $(J_\ell,J'_{\ell+1}) = U_{}^{N}[\calH_{X_1|Y_{1,\ell+1}}]$ follows from Lemma~\ref{lemma:lem1p}. 
\end{proof}
\end{appendices}

\bibliographystyle{IEEEtran}
\bibliography{IEEEabrv,bibfile}




\end{document}